\documentclass[a4paper,USenglish]{lipics-v2019}

\usepackage{algorithm}
\usepackage[noend]{algpseudocode}
\usepackage{amsmath,amssymb,amsthm}
\usepackage{graphicx}
\usepackage{hyperref}
\usepackage{subcaption} 
\usepackage{xcolor}

\newcommand{\Gtri}{\ensuremath{G_{\Delta}}}

\newcommand{\system}{\ensuremath{\mathcal{P}}}
\newcommand{\Otri}{\ensuremath{\mathcal{O}_{\Delta}}}

\newcommand{\blen}{\ensuremath{B}}
\newcommand{\bigO}[1]{\ensuremath{\mathcal{O}(#1)}}

\newcommand{\mystate}{\ensuremath{\text{state}}}
\newcommand{\parent}{\ensuremath{\text{parent}}}
\newcommand{\first}{\ensuremath{\text{first}}}
\newcommand{\bit}{\ensuremath{\text{bit}}}
\newcommand{\tokens}{\ensuremath{\text{tokens}}}
\newcommand{\plane}{\ensuremath{\text{plane}}}

\makeatletter
\algrenewcommand\ALG@beginalgorithmic{\small}
\algrenewcommand\alglinenumber[1]{\scriptsize #1:}
\makeatother

\newif\ifcomment
\commenttrue    

\newif\ifconf
\conftrue
\conffalse

\newif\ifanon
\anonfalse

\newif\iffigabbrv
\figabbrvtrue    
\figabbrvfalse   
\newcommand{\figtext}{\iffigabbrv Fig.\else Figure\fi}

\title{Convex Hull Formation for Programmable Matter}
\titlerunning{Convex Hull Formation for Programmable Matter}

\ifanon
\author{Anonymous Authors}{(Suppressed Affiliations)}{}{}{}
\authorrunning{Anonymous Authors}
\else
\author{Joshua J. Daymude}{Computer Science, CIDSE, Arizona State University, Tempe, AZ, USA}{jdaymude@asu.edu}{https://orcid.org/0000-0001-7294-5626}{}
\author{Robert Gmyr}{Department of Computer Science, University of Houston, TX, USA}{rgmyr@uh.edu}{https://orcid.org/0000-0002-2242-6083}{}
\author{Kristian Hinnenthal}{Department of Computer Science, Paderborn University, Germany}{krijan@mail.upb.de}{https://orcid.org/0000-0001-9464-295X}{}
\author{Irina Kostitsyna}{Department of Mathematics and Computer Science, TU Eindhoven, the Netherlands}{i.kostitsyna@tue.nl}{https://orcid.org/0000-0003-0544-2257}{}
\author{Christian Scheideler}{Department of Computer Science, Paderborn University, Germany}{scheidel@mail.upb.de}{}{}
\author{Andr\'ea W. Richa}{Computer Science, CIDSE, Arizona State University, Tempe, AZ, USA}{aricha@asu.edu}{}{}
\authorrunning{J.\ J.\ Daymude, R.\ Gmyr, K.\ Hinnenthal, I.\ Kostitsyna, C.\ Scheideler, and A.\ W.\ Richa}
\fi

\ifanon\else
\Copyright{Joshua J.\ Daymude, Robert Gmyr, Kristian Hinnenthal, Irina Kostitsyna, Christian Scheideler, and Andr\'ea W.\ Richa}
\fi

\ccsdesc[500]{Theory of computation~Distributed algorithms}
\ccsdesc[500]{Theory of computation~Self-organization}
\ccsdesc[300]{Theory of computation~Computational geometry}

\keywords{Programmable matter, self-organization, distributed algorithms, computational geometry, convex hull, restricted-orientation geometry}

\ifconf
\relatedversion{A full version of this paper is available online at \url{https://arxiv.org/abs/1805.06149}.}
\else\fi

\ifanon\else
\funding{J.\ J.\ Daymude and A.\ W.\ Richa gratefully acknowledge their support from the National Science Foundation under awards CCF-1422603, CCF-1637393, and CCF-1733680.
K.\ Hinnenthal and C.\ Scheideler are supported by DFG Project SCHE 1592/6-1.}
\fi

\ifconf\else
\hideLIPIcs
\nolinenumbers
\fi

\begin{document}

\maketitle
\thispagestyle{empty}

\begin{abstract}
We envision \emph{programmable matter} as a system of nano-scale agents (called \emph{particles}) with very limited computational capabilities that move and compute collectively to achieve a desired goal.
We use the \emph{geometric amoebot model} as our computational framework, which assumes particles move on the triangular lattice.
Motivated by the problem of sealing an object using minimal resources, we show how a particle system can self-organize to form an object's convex hull.
We give a distributed, local algorithm for convex hull formation and prove that it runs in $\mathcal{O}(\blen)$ asynchronous rounds, where $\blen$ is the length of the object's boundary.
Within the same asymptotic runtime, this algorithm can be extended to also form the object's (weak) $\mathcal{O}$-hull, which uses the same number of particles but minimizes the area enclosed by the hull.
Our algorithms are the first to compute convex hulls with distributed entities that have \emph{strictly local sensing, constant-size memory, and no shared sense of orientation or coordinates}.
Ours is also the first distributed approach to computing restricted-orientation convex hulls.
This approach involves coordinating particles as distributed memory; thus, as a supporting but independent result, we present and analyze an algorithm for organizing particles with constant-size memory as distributed binary counters that efficiently support increments, decrements, and zero-tests --- even as the particles move.
\end{abstract}

\newpage

\section{Introduction} \label{sec:intro}

The vision for \emph{programmable matter}, originally proposed by Toffoli and Margolus nearly thirty years ago~\cite{Toffoli1991}, is to realize a physical material that can dynamically alter its properties (shape, density, conductivity, etc.) in a programmable fashion, controlled either by user input or its own autonomous sensing of its environment.
Such systems would have broad engineering and societal impact as they could be used to create everything from reusable construction materials to self-repairing spacecraft components to even nanoscale medical devices.
While the form factor of each programmable matter system would vary widely depending on its intended application domain, a budding theoretical investigation has formed over the last decade into the algorithmic underpinnings common among these systems.
In particular, the unifying inquiry is to better understand what \emph{sophisticated, collective behaviors} are achievable by a programmable matter system composed of \emph{simple, limited computational units}.
Towards this goal, many theoretical works, complementary simulations, and even a recent experimental study~\cite{Savoie2018} have been conducted using the \emph{amoebot model}~\cite{Daymude2019} for \emph{self-organizing particle systems}.

In this paper, we give a fully local, distributed algorithm for \emph{convex hull formation} (formally defined within our context in Section~\ref{subsec:results}) under the amoebot model.
Though this well-studied problem is usually considered from the perspectives of computational geometry and combinatorial optimization as an abstraction, we treat it as the task of forming a physical seal around a static object using as few particles as possible.
This is an attractive behavior for programmable matter, as it would enable systems to, for example, isolate and contain oil spills~\cite{Zhang2013}, mimic the collective transport capabilities seen in ant colonies~\cite{McCreery2014,Kube2000}, or even surround and engulf malignant entities in the human body as phagocytes do~\cite{Aderem1999}.
Though our algorithm is certainly not the first distributed approach taken to computing convex hulls, to our knowledge it is the first to do so with distributed computational entities that have \emph{no sense of global orientation nor of their coordinates} and are limited to only \emph{local sensing and constant-size memory}.
Moreover, to our knowledge ours is the first distributed approach to computing restricted-orientation convex hulls, a generalization of usual convex hulls (see definitions in Section~\ref{subsec:amoebot}).
Finally, our algorithm has a gracefully degrading property: when the number of particles is insufficient to form an object's convex hull, a maximal partial convex hull is still formed.

\subsection{The Amoebot Model} \label{subsec:amoebot}

In the \emph{amoebot model}, originally proposed in~\cite{Derakhshandeh2014} and described in full in~\cite{Daymude2019},\footnote{See~\cite{Daymude2019} for a full motivation and description of the model, and for omitted details that were not necessary for convex hull formation.} programmable matter consists of individual, homogeneous computational elements called \emph{particles}.
Any structure that a particle system can form is represented as a subgraph of an infinite, undirected graph $G = (V,E)$ where $V$ represents all positions a particle can occupy and $E$ represents all atomic movements a particle can make.
Each node can be occupied by at most one particle.
The \emph{geometric amoebot model} further assumes $G = \Gtri$, the triangular lattice (\figtext~\ref{fig:modellattice}).

\begin{figure}[t]
\centering
\begin{subfigure}{.33\textwidth}
	\centering
	\includegraphics[scale=0.95]{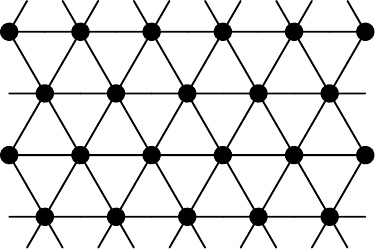}
	\caption{\centering}
	\label{fig:modellattice}
\end{subfigure}%
\begin{subfigure}{.33\textwidth}
	\centering
	\includegraphics[scale=0.95]{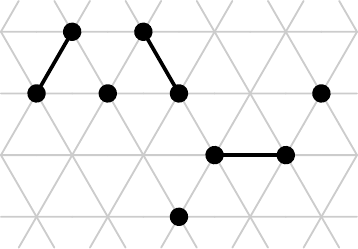}
	\caption{\centering}
	\label{fig:modelparticles}
\end{subfigure}%
\begin{subfigure}{.33\textwidth}
	\centering
	\includegraphics[scale=0.95]{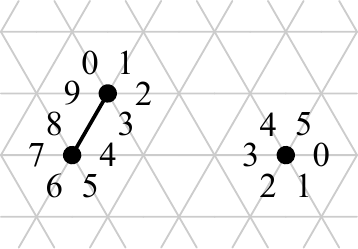}
	\caption{\centering}
	\label{fig:modellabels}
\end{subfigure}
\caption{(a) A section of the triangular lattice $\Gtri$; nodes of $V$ are shown as black circles and edges of $E$ are shown as black lines. (b) Expanded and contracted particles; $\Gtri$ is shown as a gray lattice, and particles are shown as black circles. Particles with a black line between their nodes are expanded. (c) Two particles with different offsets for their port labels.}
\label{fig:model}
\end{figure}

Each particle occupies either a single node in $V$ (i.e., it is \emph{contracted}) or a pair of adjacent nodes in $V$ (i.e., it is \emph{expanded}), as in \figtext~\ref{fig:modelparticles}.
Particles move via a series of \emph{expansions} and \emph{contractions}: a contracted particle can expand into an unoccupied adjacent node to become expanded, and completes its movement by contracting to once again occupy a single node.
An expanded particle's \emph{head} is the node it last expanded into and the other node it occupies is its \emph{tail}; a contracted particle's head and tail are both the single node it occupies.

Two particles occupying adjacent nodes are said to be \emph{neighbors}.
Neighboring particles can coordinate their movements in a \emph{handover}, which can occur in one of two ways.
A contracted particle $P$ can ``push'' an expanded neighbor $Q$ by expanding into a node occupied by $Q$, forcing it to contract.
Alternatively, an expanded particle $Q$ can ``pull'' a contracted neighbor $P$ by contracting, forcing $P$ to expand into the node it is vacating.

Each particle keeps a collection of ports --- one for each edge incident to the node(s) it occupies --- that have unique labels from its own local perspective.
Although each particle is \emph{anonymous}, lacking a unique identifier, a particle can locally identify any given neighbor by its labeled port corresponding to the edge between them.
Particles do not share a coordinate system or global compass and may have different offsets for their port labels, as in \figtext~\ref{fig:modellabels}.

Each particle has a constant-size local memory that it and its neighbors can directly read from and write to for communication.\footnote{Here, we assume the \emph{direct write communication} extension of the amoebot model as it enables a simpler description of our algorithms; see~\cite{Daymude2019} for details.}
However, particles do not have any global information and --- due to the limitation of constant-size memory --- cannot locally store the total number of particles in the system nor any estimate of this value.

The system progresses asynchronously through \emph{atomic actions}.
In the amoebot model, an atomic action corresponds to a single particle's activation, in which it can $(i)$ perform a constant amount of local computation involving information it reads from its local memory and its neighbors' memories, $(ii)$ directly write updates to at most one neighbor's memory, and $(iii)$ perform at most one expansion or contraction.
We assume these actions preserve \emph{atomicity}, \emph{isolation}, \emph{fairness}, and \emph{reliability}.
Atomicity requires that if an action is aborted before its completion (e.g., due to a conflict), any progress made by the particle(s) involved in the action is completely undone.
A set of concurrent actions preserves isolation if they do not interfere with each other; i.e., if their concurrent execution produces the same end result as if they were executed in any sequential order.
Fairness requires that each particle successfully completes an action infinitely often.
Finally, for this work, we assume these actions are reliable, meaning all particles are non-faulty.

While it is straightforward to ensure atomicity and isolation in each particle's immediate neighborhood (using a simple locking mechanism), particle writes and expansions can influence the 2-neighborhood and thus must be handled carefully.\footnote{In a manuscript in preparation, we are detailing the formal mechanisms by which atomicity and isolation can be achieved in the amoebot model~\cite{amoebot-async}.}
Conflicts of movement can occur when multiple particles attempt to expand into the same unoccupied node concurrently.
These conflicts are resolved arbitrarily such that at most one particle expands into a given node at any point in time.

It is well known that if a distributed system's actions are atomic and isolated, any set of such actions can be \emph{serialized}~\cite{Bernstein1987}; i.e., there exists a sequential ordering of the successful (non-aborted) actions that produces the same end result as their concurrent execution.
Thus, while in reality many particles may be active concurrently, it suffices when analyzing amoebot algorithms to consider a sequence of activations where only one particle is active at a time.
By our fairness assumption, if a particle $P$ is inactive at time $t$ in the activation sequence, $P$ will be (successfully) activated again at some time $t' > t$.
An \emph{asynchronous round} is complete once every particle has been activated at least once.

\paragraph*{Additional Terminology for Convex Hulls}

In addition to the formal model, we define some terminology for our application of convex hull formation, all of which are illustrated in \figtext~\ref{fig:hulls}.
An \emph{object} is a finite, static, simply connected set of nodes that does not perform computation.
The \emph{boundary} $B(O)$ of an object $O$ is the set of all nodes in $V \setminus O$ that are adjacent to $O$.
An object $O$ contains a \emph{tunnel of width $1$} if the graph induced by $V \setminus O$ is $1$-connected.
Particles are assumed to be able to differentiate between object nodes and nodes occupied by other particles.

\begin{figure}[t]
\centering
\begin{subfigure}{.48\textwidth}
	\centering
    \includegraphics[scale=.6]{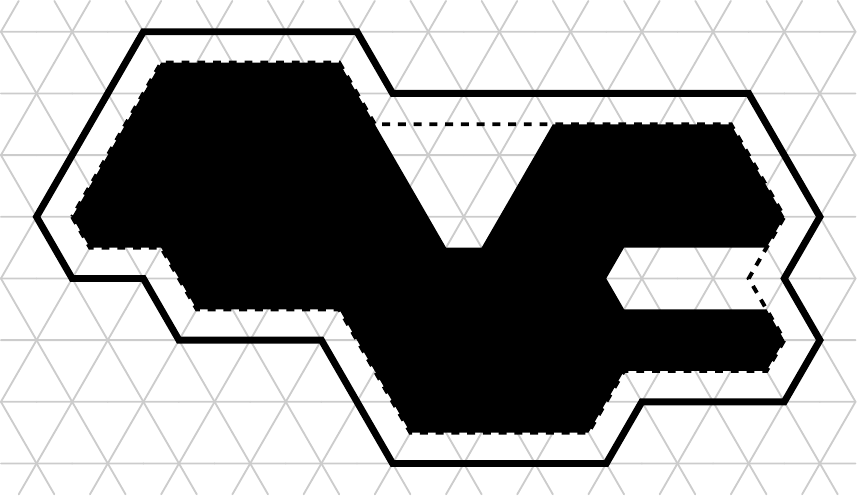}
    \caption{\centering}
    \label{fig:hulls:oconvex}
\end{subfigure}%
\begin{subfigure}{.48\textwidth}
    \centering
    \includegraphics[scale=.6]{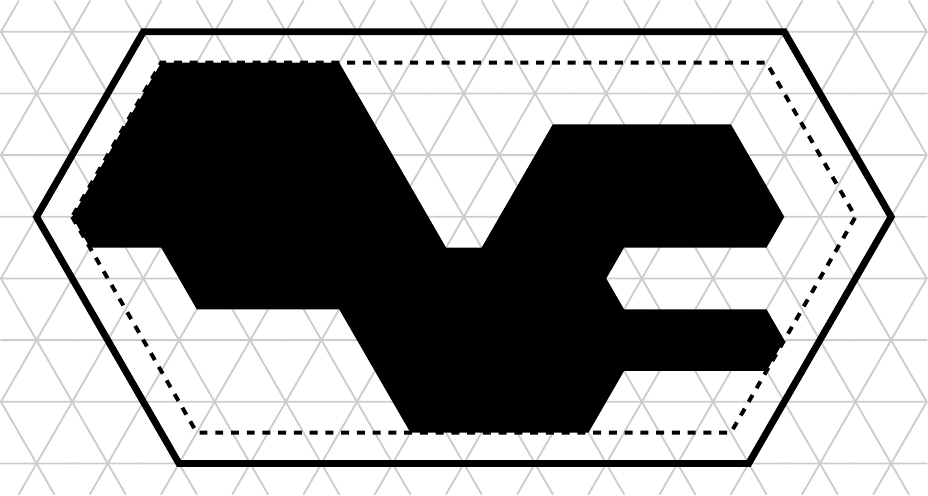}
    \caption{\centering}
    \label{fig:hulls:strongoconvex}
\end{subfigure}
\caption{An object $O$ (black) with a tunnel of width $1$ on its right side and its (a) $\mathcal{O}$-hull (dashed line) and $\Otri$-hull $H'(O)$ (solid black line), and (b) strong $\mathcal{O}$-hull (dashed line) and strong $\Otri$-hull $H(O)$ (solid black line).}
\label{fig:hulls}
\end{figure}

We now formally define the notions of convexity and convex hulls for our model.
We start by introducing the concepts of \emph{restricted-orientation convexity} (also known as \emph{$\mathcal{O}$-convexity}) and \emph{strong restricted-orientation convexity} (or \emph{strong $\mathcal{O}$-convexity}) which are well established in computational geometry~\cite{Rawlins1987,Fink2004}.
We then apply these generalized notions of convexity to our discrete setting on the triangular lattice $\Gtri$.

In the continuous setting, given a set of orientations $\mathcal{O}$ in $\mathbb{R}^2$, a geometric object is said to be \emph{$\mathcal{O}$-convex} if its intersection with every line with an orientation from $\mathcal{O}$ is either empty or connected.
The \emph{$\mathcal{O}$-hull} of a geometric object $A$ is defined as the intersection of all $\mathcal{O}$-convex sets containing $A$, or, equivalently, as the minimal $\mathcal{O}$-convex set containing $A$.
An \emph{$\mathcal{O}$-block} of two points in $\mathbb{R}^2$ is the intersection of all half-planes defined by lines with orientations in $\mathcal{O}$ and containing both points.
The \emph{strong $\mathcal{O}$-hull} of a geometric object $A$ is defined as the minimal $\mathcal{O}$-block containing $A$.

We now apply the definitions of \emph{$\mathcal{O}$-hull} and \emph{strong $\mathcal{O}$-hull} to the discrete setting of a lattice.
Let $\mathcal{O}$ be the \emph{orientation set} of $\Gtri$, i.e., the three orientations of axes of the triangular lattice.
The \emph{(weak) $\Otri$-hull} of object $O$, denoted $H'(O)$, is the set of nodes in $V \setminus O$ adjacent to the $\mathcal{O}$-hull of $O$ in $\mathbb{R}^2$ (see \figtext~\ref{fig:hulls:oconvex}).\footnote{We offset the convex hull from its traditional definition by one layer of nodes since the particles cannot occupy nodes already occupied by $O$.}
Analogously, the \emph{strong $\Otri$-hull} of object $O$, denoted $H(O)$, is the set of nodes in $V \setminus O$ adjacent to the strong $\mathcal{O}$-hull of $O$ in $\mathbb{R}^2$ (see \figtext~\ref{fig:hulls:strongoconvex}).
For simplicity, unless there is a risk of ambiguity, we will use the terms ``strong $\Otri$-hull'' and ``convex hull'' interchangeably throughout this work.

\subsection{Our Results} \label{subsec:results}

With the preceding definitions in place, we now formally define the problem we solve.
An instance of the \emph{strong $\Otri$-hull formation problem} (or the \emph{convex hull formation problem}) has the form $(\system, O)$ where $\system$ is a finite, connected system of initially contracted particles and $O \subset V$ is an object.
We assume that $(i)$ $\system$ contains a unique leader particle $\ell$ initially adjacent to $O$,\footnote{One could use the leader election algorithm for the amoebot model in~\cite{Daymude2017} to obtain such a leader in $\bigO{|\system|}$ asynchronous rounds, with high probability. Removing this assumption would simply change all the deterministic guarantees given in this work to guarantees with high probability.} $(ii)$ there are at least $|\system| > \log_2(|H(O)|)$ particles in the system, and $(iii)$ $O$ does not have any tunnels of width $1$.\footnote{We believe our algorithm could be extended to handle tunnels of width $1$ in object $O$, but this would require technical details beyond the scope of this paper.}
A local, distributed algorithm $\mathcal{A}$ solves an instance $(\system, O)$ of the convex hull formation problem if, when each particle executes $\mathcal{A}$ individually, $\system$ is reconfigured so that every node of $H(O)$ is occupied by a contracted particle.
The \emph{$\Otri$-hull formation problem} can be stated analogously.

Let $B = |B(O)|$ denote the length of the object's boundary and $H = |H(O)|$ denote the length of the object's convex hull.
We present a \textbf{local, distributed algorithm} for the \textbf{strong $\Otri$-hull formation problem} that runs in \textbf{$\mathcal{O}(B)$ rounds} and later show how it can be extended to also solve the \textbf{$\Otri$-hull formation problem} in an additional \textbf{$\mathcal{O}(H)$ rounds}.
Our algorithm has a gracefully degrading property: if there are insufficient particles to completely fill the strong $\Otri$-hull with contracted particles --- i.e., if $|\system| < H$ --- our algorithm will still form a maximal partial strong $\Otri$-hull.
To our knowledge, our algorithm is the first to address distributed convex hull formation using nodes that have no sense of global orientation nor of their coordinates and are limited to only constant-size memory and local communication.
It is also the first distributed algorithm for forming restricted-orientation convex hulls (see Section~\ref{subsec:relwork}).

Our approach critically relies on the particle system maintaining and updating the distances between the leader's current position and each of the half-planes whose intersection composes the object's convex hull.
However, these distances can far exceed the memory capacity of an individual particle; the former can be linear in the perimeter of the object, while the latter is constant.
To address this problem, we give new results on coordinating a particle system as a distributed binary counter that supports increments and decrements by one as well as testing the counter value's equality to zero (i.e., \emph{zero-testing}).
These new results supplant preliminary work on increment-only distributed binary counters under the amoebot model~\cite{Porter2018}, and we stress that this extension is non-trivial.
Moreover, these results are agnostic of convex hull formation and can be used as a modular primitive for future applications.

\subsection{Related Work} \label{subsec:relwork}

The convex hull problem is arguably one of the best-studied problems in computational geometry.
Many parallel algorithms have been proposed to solve it (see, e.g.,~\cite{Akl1993,Fjallstrom1990,Dymond2001}), as have several distributed algorithms (see~\cite{Rajsbaum2011,Miller1988,Diallo1999}).
However, conventional models of parallel and distributed computation assume that the computational and communication capabilities of the individual processors far exceed those of individual particles of programmable matter.
Most commonly, for example, the nodes are assumed to know their global coordinates and to can communicate non-locally.
Particles in the amoebot model have only constant-size memory and can communicate only with their immediate neighbors.
Furthermore, the object's boundary may be much larger than the number of particles, making it impossible for the particle system to store all the geographic locations.
Finally, to our knowledge, there only exist centralized algorithms to compute (strong) restricted-orientation convex hulls (see, e.g.,~\cite{Karlsson1988} and the references therein); ours is the first to do so in a distributed setting.

The amoebot model for self-organizing particle systems can be classified as an \emph{active} system of programmable matter --- in which the computational units have control over their own movements and actions --- as opposed to \emph{passive} systems such as population protocols and models of molecular self-assembly (see, e.g.,~\cite{Angluin2006,Patitz2014}).
Other active systems include modular self-reconfigurable robot systems (see, e.g.,~\cite{Yim2007,Gilpin2010} and the references therein) and the nubot model for molecular computing by Woods et al.~\cite{Woods2013}.
One might also include the mobile robots model in this category (see~\cite{Flocchini2019} and the references therein), in which robots abstracted as points in the real plane or on graphs solve problems such as pattern formation and gathering.
A notable difference between the amoebot model and the mobile robots literature is in their treatment of progress and time: mobile robots progress according to fine-grained ``look-compute-move'' cycles where actions are comprised of exactly one look, move, or compute operation.
In comparison --- at the scale where particles can only perform a constant amount of computation and are restricted to immediate neighborhood sensing --- the amoebot model assumes coarser atomic actions (as described in Section~\ref{subsec:amoebot}).



Lastly, we briefly distinguish convex hull formation from the related problems of shape formation and object coating, both of which have been considered under the amoebot model.
Like shape formation~\cite{Derakhshandeh2016}, convex hull formation is a task of reconfiguring the particle system's shape; however, the desired hull shape is based on the object and thus is not known to the particles ahead of time.
Object coating~\cite{Derakhshandeh2017} also depends on an object, but may not form a convex seal around the object using the minimum number of particles.

\subsection{Organization} \label{subsec:organization}

Our convex hull formation algorithm has two phases: the particle system first explores the object to learn the convex hull's dimensions, and then uses this knowledge to form the convex hull.
In Section~\ref{sec:soloalg}, we introduce the main ideas behind the learning phase as a novel local algorithm run by a single particle with unbounded memory.
We then give new results on organizing a system of particles each with $\bigO{1}$ memory into binary counters in Section~\ref{sec:counter}.
Combining the results of these two sections, we present the full distributed algorithm for learning and forming the strong $\Otri$-hull in Section~\ref{sec:multialg}.
We conclude by presenting an extension of our algorithm to solve the $\Otri$-hull formation problem in Section~\ref{sec:weakhull}.

\section{The Single-Particle Algorithm} \label{sec:soloalg}

We first consider a particle system composed of a single particle $P$ with unbounded memory and present a local algorithm for learning the strong $\Otri$-hull of object $O$.
As will be the case in the distributed algorithm, particle $P$ does not know its global coordinates or orientation.
We assume $P$ is initially on $B(O)$, the boundary of $O$.
The main idea of this algorithm is to let $P$ perform a clockwise traversal of $B(O)$, updating its knowledge of the convex hull as it goes.

In particular, the convex hull can be represented as the intersection of six half-planes $\mathcal{H} = \{N, NE, SE, S, SW, NW\}$, which $P$ can label using its local compass (see \figtext~\ref{fig:halfplanes}).
Particle $P$ estimates the location of these half-planes by maintaining six counters $\{d_h : h \in \mathcal{H}\}$, where each counter $d_h$ represents the $L_1$-distance\footnote{The distance from a node to a half-plane is the number of edges in a shortest path between the node and any node on the line defining the half-plane.} from the position of $P$ to half-plane $h$.
If at least one of these counters is equal to $0$, $P$ is on its current estimate of the convex hull.

\begin{figure}[t]
\centering
\begin{subfigure}{.45\textwidth}
    \centering
    \includegraphics[scale=1.3]{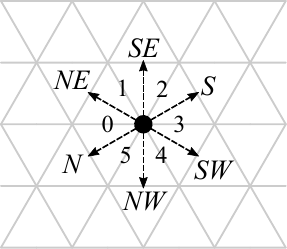}
    \caption{\centering}
    \label{fig:halfplanes:compass}
\end{subfigure}%
\begin{subfigure}{.45\textwidth}
    \centering
    \includegraphics[scale=0.6]{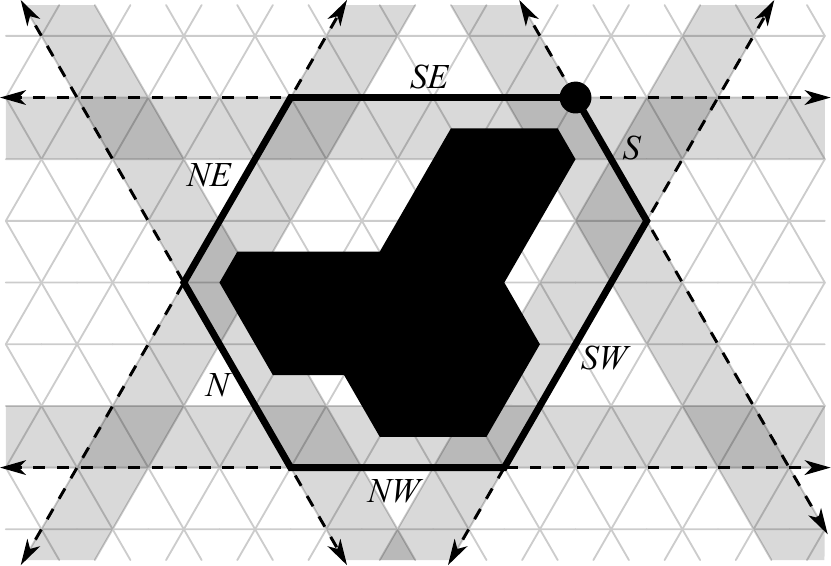}
    \caption{\centering}
    \label{fig:halfplanes:hull}
\end{subfigure}
\caption{(a) A particle's local labeling of the six half-planes composing the convex hull: the half-plane between its local $0$ and $5$-labeled edges is $N$, and the remaining half-planes are labeled accordingly. (b) An object (black) and the six half-planes (dashed lines with shading on included side) whose intersection forms its convex hull (black line). As an example, the node depicted in the upper-right is distance $0$ from the $S$ and $SE$ half-planes and distance $7$ from $N$.}
\label{fig:halfplanes}
\end{figure}

Each counter is initially set to $0$, and $P$ updates them as it moves.
Let $[6] = \{0, \ldots, 5\}$ denote the six directions $P$ can move in, corresponding to its contracted port labels.
In each step, $P$ first computes the direction $i \in [6]$ to move toward using the right-hand rule, yielding a clockwise traversal of $B(O)$.
Since $O$ was assumed to not have tunnels of width $1$, direction $i$ is unique.
Particle $P$ then updates its distance counters by setting $d_h \gets \max\{0, d_h + \delta_{i,h}\}$ for all $h \in \mathcal{H}$, where $\delta_i = (\delta_{i,N}, \delta_{i,NE}, \delta_{i,SE}, \delta_{i,S}, \delta_{i,SW}, \delta_{i,NW})$ is defined as follows:
\[\begin{array}{ccc}
\delta_0 = (1, 1, 0, -1, -1, 0)~~~~~ & \delta_1 = (0, 1, 1, 0, -1, -1)~~~~~ & \delta_2 = (-1, 0, 1, 1, 0, -1) \\
\delta_3 = (-1, -1, 0, 1, 1, 0)~~~~~ & \delta_4 = (0, -1, -1, 0, 1, 1)~~~~~ & \delta_5 = (1, 0, -1, -1, 0, 1)
\end{array}\]

Thus, every movement decrements the distance counters of the two half-planes to which $P$ gets closer, and increments the distance counters of the two half-planes from which $P$ gets farther away.
Whenever $P$ moves toward a half-plane to which its distance is already $0$, the value stays $0$, essentially ``pushing'' the estimation of the half-plane one step further.
An example of such a movement is given in \figtext~\ref{fig:estimate}.

\begin{figure}[t]
\centering
\begin{subfigure}{.33\textwidth}
	\centering
	\includegraphics[scale=0.75]{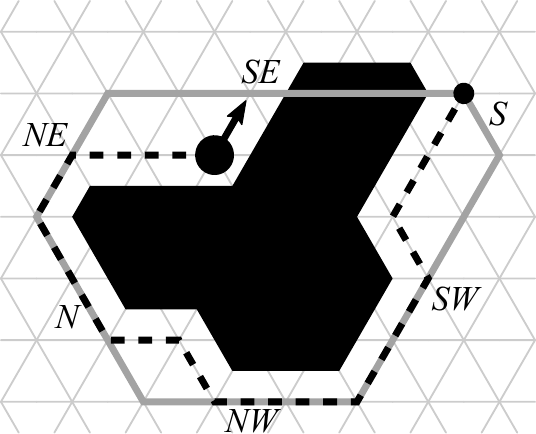}
	\caption{\centering}
	\label{fig:estimate:a}
\end{subfigure}%
\hfill
\begin{subfigure}{.33\textwidth}
	\centering
	\includegraphics[scale=0.75]{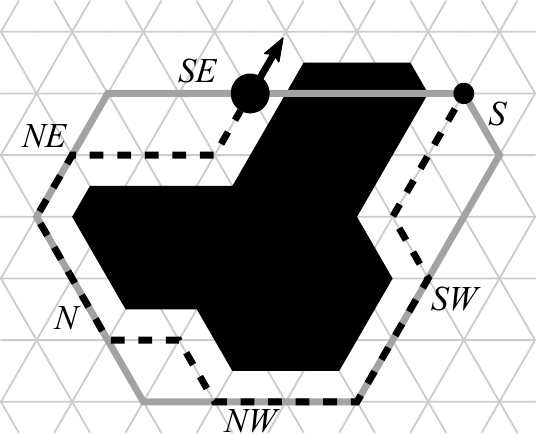}
	\caption{\centering}
	\label{fig:estimate:b}
\end{subfigure}%
\hfill
\begin{subfigure}{.33\textwidth}
	\centering
	\includegraphics[scale=0.75]{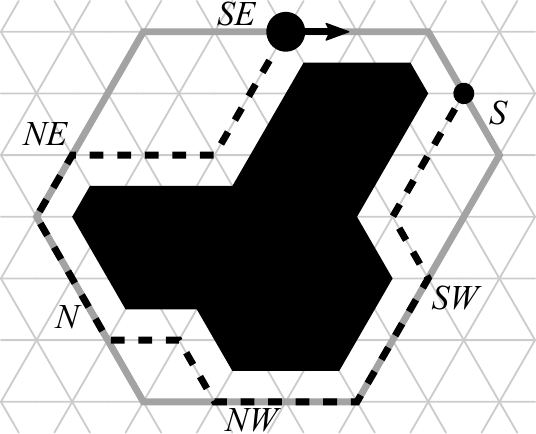}
	\caption{\centering}
	\label{fig:estimate:c}
\end{subfigure}
\caption{The particle $P$ with its convex hull estimate (gray line) after traversing the path (dashed line) from its starting point (small black dot). (a) $d_h \geq 1$ for all $h \in \mathcal{H}$, so its next move does not push a half-plane. (b) Its next move is toward the $SE$ half-plane and $d_{SE} = 0$, so (c) $SE$ is pushed.}
\label{fig:estimate}
\end{figure}

Finally, $P$ needs to detect when it has learned the complete convex hull.
To do so, it stores six terminating bits $\{b_h : h \in \mathcal{H}\}$, where $b_h$ is equal to $1$ if $P$ has visited half-plane $h$ (i.e., if $d_h$ has been $0$) since $P$ last pushed any half-plane, and $0$ otherwise.
Whenever $P$ moves without pushing a half-plane (e.g., \figtext~\ref{fig:estimate:a}--\ref{fig:estimate:b}), it sets $b_h = 1$ for all $h$ such that $d_h = 0$ after the move.
If its move pushed a half-plane (e.g., \figtext~\ref{fig:estimate:b}--\ref{fig:estimate:c}), it resets all its terminating bits to $0$.
Once all six terminating bits are $1$, $P$ contracts and terminates.

\paragraph*{Analysis} \label{subsec:soloanalysis}

We now analyze the correctness and runtime of this single-particle algorithm.
\ifconf\else
Note that, since the particle system contains only one particle $P$, each activation of $P$ is also an asynchronous round.
\fi
For a given round $i$, let $H_i(O) \subset V$ be the set of all nodes enclosed by $P$'s estimate of the convex hull of $O$ after round $i$, i.e., all nodes in the closed intersection of the six half-planes.
We first show that $P$'s estimate of the convex hull represents the correct convex hull $H(O)$ after at most one traversal of the object's boundary, and does not change afterwards.

\begin{lemma} \label{lem:estimate}
If particle $P$ completes its traversal of $B(O)$ in round $i^*$, then $H_i(O) = H(O)$ for all $i \geq i^*$.
\end{lemma}
\ifconf\else
\begin{proof}
Since $P$ exclusively traverses $B(O)$, $H_i(O) \subseteq H(O)$ for all rounds $i$.
Furthermore, $H_i(O) \subseteq H_{i+1}(O)$ for any round $i$.
Once $P$ has traversed the whole boundary, it has visited a node of each half-plane corresponding to $H(O)$, and thus $H_{i^*}(O) = H(O)$.
\end{proof}
\fi

\ifconf\else
We now show particle $P$ terminates if and only if it has learned the complete convex hull.
\fi

\begin{lemma} \label{lem:earlyterm}
If $H_i(O) \subset H(O)$ after some round $i$, then $b_h = 0$ for some half-plane $h \in \mathcal{H}$.
\end{lemma}
\ifconf\else
\begin{proof}
Suppose to the contrary that after round $i$, $H_i(O) \subset H(O)$ but $b_h = 1$ for all $h \in \mathcal{H}$; let $i$ be the first such round.
Then after round $i-1$, there was exactly one half-plane $h_1 \in \mathcal{H}$ such that $b_{h_1} = 0$; all other half-planes $h \in \mathcal{H} \setminus \{h_1\}$ have $b_h = 1$.
Let $h_2, \ldots, h_6$ be the remaining half-planes in clockwise order, and let round $t_j < i-1$ be the one in which $b_{h_j}$ was most recently flipped from $0$ to $1$, for $2 \leq j \leq 6$.
Particle $P$ could only set $b_{h_j} = 1$ in round $t_j$ if its move in round $t_j$ did not push any half-planes and $d_{h_j} = 0$ after the move.
There are two ways this could have occurred.

First, $P$ may have already had $d_{h_j} = 0$ in round $t_j - 1$ and simply moved along $h_j$ in round $t_j$, leaving $d_{h_j} = 0$.
But for this to hold and for $P$ to have had $b_{h_j} = 0$ after round $t_j - 1$, $P$ must have just pushed $h_j$, resetting all its terminating bits to $0$.
Particle $P$ could not have pushed any half-plane during rounds $t_2$ up to $i-1$, since $b_{h_2} = \cdots = b_{h_6} = 1$, so this case only could have occurred with half-plane $h_2$.

For the remaining half-planes $h_j$, for $3 \leq j \leq 6$, $P$ must have had $d_{h_j} = 1$ in round $t_j - 1$ and moved into $h_j$ in round $t_j$.
But this is only possible if $P$ pushed $h_j$ in some round prior to $t_j - 1$, implying that $P$ has already visited $h_3, \ldots, h_6$.
Therefore, $P$ has completed at least one traversal of $B(O)$ by round $i$, but $H_i(O) \subset H(O)$, contradicting Lemma~\ref{lem:estimate}.
\end{proof}
\fi

\begin{lemma} \label{lem:termination}
Suppose $H_i(O) = H(O)$ for the first time after some round $i$. Then particle $P$ terminates at some node of $H(O)$ after at most one additional traversal of $B(O)$.
\end{lemma}
\ifconf\else
\begin{proof}
Since $i$ is the first round in which $H_i(O) = H(O)$, particle $P$ must have just pushed some half-plane $h$ --- resetting all its terminating bits to $0$ --- and now occupies a node $u$ with distance $0$ to $h$.
Due to the geometry of the triangular lattice, the next node in a clockwise traversal of $B(O)$ from $u$ must also have distance $0$ to $h$, so $P$ will set $b_h$ to $1$ after its next move.
As $P$ continues its traversal, it will no longer push any half-planes because its convex hull estimation is complete.
Thus, $P$ will visit every other half-plane $h'$ without pushing it, causing $P$ to set each $b_{h'}$ to $1$ before reaching $u$ again.
Particle $P$ sets its last terminating bit $b_{h^*}$ to $1$ when it next visits a node $v$ with distance $0$ to $h^*$.
Therefore, $P$ terminates at $v \in B(O) \cap H(O)$.
\end{proof}
\fi

The previous lemmas immediately imply the following theorem.
Let $\blen = |B(O)|$.

\begin{theorem} \label{thm:singlecorrect}
The single-particle algorithm terminates after $t^* = \bigO{\blen}$ asynchronous rounds with particle $P$ at a node $u \in B(O) \cap H(O)$ and $H_{t^*}(O) = H(O)$.
\end{theorem}

\section{A Binary Counter of Particles} \label{sec:counter}

For a system of particles each with constant-size memory to emulate the single-particle algorithm of Section~\ref{sec:soloalg}, the particles need a mechanism to distributively store the distances to each of the strong $\Otri$-hull's six half-planes.
To that end, we now describe how to coordinate a particle system as a binary counter that supports increments and decrements by one as well as zero-testing.
Accompanying pseudocode can be found in Appendix~\ref{subapp:countercode}.
This description subsumes previous work on collaborative computation under the amoebot model that detailed an increment-only binary counter~\cite{Porter2018}.
This algorithm uses \emph{tokens}, or constant-size messages that can be passed between particles~\cite{Daymude2019}.

Suppose that the participating particles are organized as a simple path with the leader particle at its start: $\ell = P_0, P_1, P_2, \ldots, P_k$.
Each particle $P_i$ stores a value $P_i.\bit \in \{\emptyset, 0, 1\}$, where $P_i.\bit = \emptyset$ implies $P_i$ is not part of the counter; i.e., it is beyond the most significant bit.
Each particle $P_i$ also stores tokens in a queue $P_i.\tokens$; the leader $\ell$ can only store one token, while all other particles can store up to two.
These tokens can be increments $c^+$, decrements $c^-$, or the unique \emph{final token} $f$ that represents the end of the counter.
If a particle $P_i$ (for $0 < i \leq k$) holds $f$ --- i.e., $P_i.\tokens$ contains $f$ --- then the counter value is represented by the bits of each particle from the leader $\ell$ (storing the least significant bit) up to and including $P_{i-1}$ (storing the most significant bit).

The leader $\ell$ is responsible for initiating counter operations, while the rest of the particles use only local information and communication to carry these operations out.
To increment the counter, the leader $\ell$ generates an increment token $c^+$ (assuming it was not already holding a token).
Now consider this operation from the perspective of any particle $P_i$ holding a $c^+$ token, where $0 \leq i \leq k$.
If $P_i.\bit = 0$, $P_i$ consumes $c^+$ and sets $P_i.\bit \gets 1$.
Otherwise, if $P_i.\bit = 1$, this increment needs to be carried over to the next most significant bit.
As long as $P_{i+1}.\tokens$ is not full (i.e., $P_{i+1}$ holds at most one token), $P_i$ passes $c^+$ to $P_{i+1}$ and sets $P_i.\bit \gets 0$.
Finally, if $P_i.\bit = \emptyset$, this increment has been carried over past the counter's end, so $P_i$ must also be holding the final token $f$.
In this case, $P_i$ forwards $f$ to $P_{i+1}$, consumes $c^+$, and sets $P_i.\bit \gets 1$.

To decrement the counter, the leader $\ell$ generates a decrement token $c^-$ (if it was not holding a token).
From the perspective of any particle $P_i$ holding a $c^-$ token, where $0 \leq i < k$, the cases for $P_i.\bit \in \{0,1\}$ are nearly anti-symmetric to those for the increment.
If $P_i.\bit = 0$ and $P_{i+1}.\tokens$ is not full, $P_i$ carries this decrement over by passing $c^-$ to $P_{i+1}$ and setting $P_i.\bit \gets 1$.
However, if $P_i.\bit = 1$, we only allow $P_i$ to consume $c^-$ and set $P_i.\bit \gets 0$ if $P_{i+1}.\bit \neq 1$ or $P_{i+1}$ is not only holding a $c^-$.
While not necessary for the correctness of the decrement operation, this will enable conclusive zero-testing.
Additionally, if $P_{i+1}$ is holding $f$, then $P_i$ is the most significant bit.
So this decrement shrinks the counter by one bit; thus, as long as $P_i \neq P_0$, $P_i$ additionally takes $f$ from $P_{i+1}$, consumes $c^-$, and sets $P_i.\bit \gets \emptyset$.

Finally, the zero-test operation: if $P_1.\bit = 1$ and $P_1$ only holds a decrement token $c^-$, $\ell$ cannot perform the zero-test conclusively (i.e., zero-testing is ``unavailable'').
Otherwise, the counter value is $0$ if and only if $P_1$ is only holding the final token $f$ and $(i)$ $\ell.\bit = 0$ and $\ell.\tokens$ is empty or $(ii)$ $\ell.\bit = 1$ and $\ell$ is only holding a decrement token $c^-$.

\subsection{Correctness} \label{subsec:countercorrect}

We now show the \emph{safety} of our increment, decrement, and zero-test operations for the distributed counter.
More formally, we show that given any sequence of these operations, our distributed binary counter will eventually yield the same values as a centralized counter, assuming the counter's value remains nonnegative.

If our distributed counter was fully synchronized, meaning at most one increment or decrement token is in the counter at a time, the distributed counter would exactly mimic a centralized counter but with a linear slowdown in its length.
Our counter instead allows for many increments and decrements to be processed in a pipelined fashion.
Since the $c^+$ and $c^-$ tokens are prohibited from overtaking one another, thereby altering the order the operations were initiated in, it is easy to see that the counter will correctly process as many tokens as there is capacity for.

So it remains to prove the correctness of the zero-test operation.
We will prove this in two parts: first, we show the zero-test operation is always eventually available.
We then show that if the zero-test operation is available, it is always reliable; i.e., it always returns an accurate indication of whether or not the counter's value is $0$.

\begin{lemma} \label{lem:0testavailable}
If at time $t$ zero-testing is unavailable (i.e., particle $P_1$ is holding a decrement token $c^-$ and $P_1.\bit = 1$) then there exists a time $t' > t$ when zero-testing is available.
\end{lemma}
\begin{proof}
\ifconf
\else
We argue by induction on $i$ --- the number of consecutive particles $P_j$, for $x \leq j < x+i$ , such that $P_j.\bit = 1$ and $P_j$ only holds a $c^-$ token --- that there exists a time $t^* > t$ where $P_x$ can consume $c^-$ and set $P_x.\bit \gets 0$.
If $i = 1$, then either $P_{x+1}.\bit \neq 1$ or $P_{i+1}$ is not only holding a $c^-$, so $P_x$ can process its $c^-$ at its next activation (say, at $t^* > t$).

Now suppose $i > 1$ and the induction hypothesis holds up to $i-1$.
Then at time $t$, every particle $P_j$ with $x \leq j < x+i-1$ is holding a $c^-$ token and has $P_j.\bit = 1$.
By the induction hypothesis, there exists a time $t_1 > t$ at which $P_{x+1}$ is activated and can consume its $c^-$ token, setting $P_{x+1}.\bit \gets 0$.
So the next time $P_x$ is activated (say, at $t^* > t_1$) it can do the same, consuming its $c^-$ token and setting $P_x.\bit \gets 0$.
This concludes our induction.

Suppose $P_1$ is holding a decrement token $c^-$ and $P_1.\bit = 1$ at time $t$, leaving the zero-test unavailable.
Applying the above argument to $P_1$, there must exist a time $t^* > t$ such that $P_1$ can process its $c^-$ and set $P_1.\bit \gets 0$.
Since the increment and decrement tokens remain in order, $P_1$ will not be holding a $c^-$ token when $\ell$ is next activated (say, at $t' > t^*$) allowing $\ell$ to perform a zero-test.
\fi
\end{proof}

\begin{lemma} \label{lem:0testreliable}
If the zero-test operation is available, then it reliably decides whether the counter's value is $0$.
\end{lemma}
\begin{proof}
The statement of the lemma can be rephrased as follows: assuming the zero-test operation is available, the value of the counter $v = 0$ if and only if $P_1$ only holds the final token $f$ and either $(i)$ $\ell.\bit = 0$ and $\ell.\tokens$ is empty or $(ii)$ $\ell.\bit = 1$ and $\ell$ only holds a decrement token $c^-$.
Let $(*)$ represent the right hand side of this iff.
Note that $v$ is defined only in terms of the operations the leader has initiated, not in terms of what the particles have processed.

We first prove the reverse direction: if $(*)$ holds, then $v = 0$.
By $(*)$, we know that $P_1.\tokens$ only holds $f$.
Thus, $\ell.\bit$ is both the least significant bit (LSB) and the most significant bit (MSB).
Also by $(*)$ we know that either $\ell.\bit = 0$ and $\ell.\tokens$ is empty, or $\ell.\bit = 1$ and $\ell.\tokens = [c^-]$.
In either case, it is easy to see that $v = 0$.

To prove that if $v = 0$, then $(*)$ holds, we argue by induction on the number of operations $i$ initiated by the leader (i.e., the total number of $c^+$ and $c^-$ tokens generated by $\ell$).
Initially, no operations have been initiated, so $v = 0$.
The counter is thus in its initial configuration: $P_1.\tokens$ only contains $f$, $\ell.\bit = 0$, and $\ell.\tokens$ is empty.
So $(*)$ holds.
Now suppose that the induction hypothesis holds for the first $i - 1$ operations initiated, and consider the time $t_{i-1}$ just before $\ell$ generates the $i$-th operation at time $t_i$.
There are two cases to consider: at time $t_{i-1}$, either $v = 0$ or $v > 0$.

Suppose $v = 0$ at time $t_{i-1}$.
Since $\ell$ can only hold one token, $\ell.\tokens$ must have been empty at time $t_{i-1}$ in order for $\ell$ to initiate another operation at time $t_i$.
This operation must have been an increment, since a decrement on $v = 0$ violates the counter's nonnegativity.
So at time $t_i$, $v = 1 > 0$ and thus ``if $v = 0$, then $(*)$ holds'' is vacuously true.

So suppose $v > 0$ at time $t_{i-1}$.
The only nontrivial case is when $v = 1$ at time $t_{i-1}$ and the $i$-th operation is a decrement; otherwise, $v$ remains greater than $0$ and ``if $v = 0$, then $(*)$ holds'' is vacuously true.
In this nontrivial case, $v = 0$ and $\ell.\tokens = [c^-]$ at time $t_i$.
To show $(*)$ holds, we must establish that $\ell.\bit = 1$ and $P_1$ only holds $f$ at time $t_i$.
Suppose to the contrary that $\ell.\bit = 0$ at time $t_i$.
Then the $c^-$ token in $\ell.\tokens$ must eventually be carried over to some particle $P_j$ with $j \geq 1$ that will process it.
But this implies that $v > 2^j - 1 \geq 1$ at time $t_i$, a contradiction that $v = 0$.

Finally, suppose to the contrary that $P_1.\tokens \neq [f]$ at time $t_i$.
If $P_1.\bit = \emptyset$, we reach a contradiction because $\ell.\bit = 0$ is the LSB and $\ell.\tokens = [c^-]$, implying that $v < 0$.
If $P_1.\bit = 0$, we reach a contradiction because $\ell.\bit = P_1.\bit = 0$ and thus there must exist a particle $P_j$ with $j \geq 2$ that will consume the $c^-$ token held by $\ell$, implying that $v > 2^j - 1 \geq 3$.
So we have that $P_1.\bit = 1$ at time $t_i$.
If $P_1.\tokens = [c^-]$, we reach a contradiction because the zero-test operation is available.
If $P_1.\tokens$ is empty or contains a $c^+$ token, we reach a contradiction because $P_1.\bit = 1$, implying that $v > 1$.
But since $P_1$ cannot hold two $c^-$ tokens (as $\ell$ would had to have consumed a previous $c^-$ token while $P_1.\bit = 1$ and $P_1.\tokens = [c^-]$) and cannot hold both $f$ and a $c^-$ token (as this implies $v < 0$), the only remaining case is that $P_1.\tokens = [f]$, a contradiction.
\end{proof}

\subsection{Runtime} \label{subsec:counterruntime}

To analyze the runtime of our distributed binary counters, we use a \emph{dominance argument} between asynchronous and parallel executions, building upon the analysis of~\cite{Porter2018} that bounded the running time of an increment-only distributed counter.
The general idea of the argument is as follows.
First, we prove that the counter operations are, in the worst case, at least as fast in an asynchronous execution as they are in a simplified parallel execution.
We then give an upper bound on the number of parallel rounds required to process these operations; combining these two results also gives a worst case upper bound on the running time in terms of asynchronous rounds.

Let a configuration $C$ of the distributed counter encode each particle's bit value and any increment or decrement tokens it might be holding.
A configuration is \emph{valid} if $(i)$ there is exactly one particle (say, $P_i$) holding the final token $f$, $(ii)$ $P_j.\bit = \emptyset$ if $j \geq i$ and $P_j.\bit \in \{0,1\}$ otherwise, and $(iii)$ if a particle $P_j$ is holding a $c^+$ or $c^-$ token, then $j \leq i$.
A \emph{schedule} is a sequence of configurations $(C_0, \ldots, C_t)$.
Let $S$ be a \emph{nonnegative} sequence of $m$ increment and decrement operations; i.e., for all $0 \leq i \leq m$, the first $i$ operations have at least as many increments as decrements.

\begin{definition} \label{def:counterschedule}
A \emph{parallel counter schedule} $(S, (C_0, \ldots, C_t))$ is a schedule $(C_0, \ldots, C_t)$ such that each configuration $C_i$ is valid, each particle holds at most one token, and, for every $0 \leq i < t$, $C_{i+1}$ is reached from $C_i$ by satisfying the following for each particle $P_j$:
\begin{enumerate}
\item If $j = 0$, then $P_j = \ell$ generates the next operation according to $S$.
\item $P_j$ is holding $c^+$ in $C_i$ and either $P_j.\bit = 0$, causing $P_j$ to consume $c^+$ and set $P_j.\bit \gets 1$, or $P_j.\bit = \emptyset$, causing $P_j$ to additionally pass the final token $f$ to $P_{j+1}$.
\item $P_j$ is holding $c^-$ and $P_j.\bit = 1$ in $C_i$, so $P_j$ consumes $c^-$. If $P_{j+1}$ is holding $f$ in $C_i$, $P_j$ takes $f$ from $P_{j+1}$ and sets $P_j.\bit \gets \emptyset$; otherwise it simply sets $P_j.\bit \gets 0$.
\item $P_j$ is holding $c^+$ and $P_j.\bit = 1$ in $C_i$, so $P_j$ passes $c^+$ to $P_{j+1}$ and sets $P_j.\bit \gets 0$.
\item $P_j$ is holding $c^-$ and $P_j.\bit = 0$ in $C_i$, so $P_j$ passes $c^-$ to $P_{j+1}$ and sets $P_j.\bit \gets 1$.
\end{enumerate}
Such a schedule is said to be \emph{greedy} if the above actions are taken whenever possible.
\end{definition}

Using the same sequence of operations $S$ and a fair asynchronous activation sequence $A$, we compare a greedy parallel counter schedule to an \emph{asynchronous counter schedule} $(S, (C_0^A, \ldots, C_t^A))$, where $C_i^A$ is the resulting configuration after asynchronous round $i$ completes according to $A$.
Recall that in the asynchronous setting, each particle (except the leader $\ell$) is allowed to hold up to two counter tokens at once while the parallel schedule is restricted to at most one token per particle (Definition~\ref{def:counterschedule}).
For a given (increment or decrement) token $c$, let $I_C(c)$ be the index of the particle holding $c$ in configuration $C$ if such a particle exists, or $\infty$ if $c$ has already been consumed.
For any two configurations $C$ and $C'$ and any token $c$, we say $C$ \emph{dominates} $C'$ \emph{with respect to} $c$ --- denoted $C(c) \succeq C'(c)$ --- if and only if $I_C(c) \geq I_{C'}(c)$.
We say $C$ \emph{dominates} $C'$ --- denoted $C \succeq C'$ --- if and only if $C(c) \succeq C'(c)$ for every token $c$.

\begin{lemma} \label{lem:counterdominance}
Given any nonnegative sequence of operations $S$ and any fair asynchronous activation sequence $A$ beginning at a valid configuration $C_0^A$ in which each particle holds at most one token, there exists a greedy parallel counter schedule $(S, (C_0, \ldots, C_t))$ with $C_0 = C_0^A$ such that $C_i^A \succeq C_i$ for all $0 \leq i \leq t$.
\end{lemma}
\ifconf\else
\begin{proof}
With a nonnegative sequence of operations $S$, a fair activation sequence $A$, and a valid starting configuration $C_0^A$, we obtain a unique asynchronous counter schedule $(S, (C_0^A, \ldots, C_t^A))$.
We construct a greedy parallel counter schedule $(S, (C_0, \ldots, C_t))$ using the same sequence of operations $S$ as follows.
Let $C_0 = C_0^A$, and note that since each particle in $C_0^A$ was assumed to hold at most one token, $C_0$ is a valid parallel configuration.
Next, for $0 \leq i < t$, let $C_{i+1}$ be obtained from $C_i$ by performing one \emph{parallel round}: each particle greedily performs one of Actions 2--5 of Definition~\ref{def:counterschedule} if possible; the leader $\ell$ additionally performs Action 1 if possible.

To show $C_i^A \succeq C_i$ for all $0 \leq i \leq t$, argue by induction on $i$.
Clearly, since $C_0 = C_0^A$, we have $I_{C_0}(c) = I_{C_0^A}(c)$ for any token $c$ in the counter.
Thus, $C_0^A \succeq C_0$.
So suppose by induction that for all rounds $0 \leq r < i$, we have $C_r^A \succeq C_r$.
Consider any counter token $c$ in $C_i$.
Since both the asynchronous and parallel schedules follow the same sequence of operations $S$, it suffices to show that $I_{C_i}(c) \leq I_{C_i^A}(c)$.
By the induction hypothesis, we have that $I_{C_{i-1}}(c) \leq I_{C_{i-1}^A}(c)$, but there are two cases to distinguish between:

\medskip

\noindent\textsf{Case 1.} Token $c$ has made strictly more progress in the asynchronous setting than in the parallel setting by round $i-1$, i.e., $I_{C_{i-1}}(c) < I_{C_{i-1}^A}(c)$.
If $c$ is consumed in parallel round $i$, then $c$ must have been consumed at some time before asynchronous round $i$.
Otherwise, since $c$ is carried over at most once per parallel round, $I_{C_i}(c) \leq I_{C_{i-1}}(c) + 1 \leq I_{C_{i-1}^A}(c) \leq I_{C_i^A}(c)$.

\smallskip

\noindent\textsf{Case 2.} Token $c$ has made the same amount of progress in the asynchronous and parallel settings by round $i-1$, i.e., $I_{C_{i-1}}(c) = I_{C_{i-1}^A}(c)$.
Inspection of Definition~\ref{def:counterschedule} shows that nothing can block $c$ from making progress in the next parallel round, a fact we will formalize in Lemma~\ref{lem:counterprogress}.
So if $c$ is consumed in parallel round $i$, we must show it is also consumed in asynchronous round $i$; otherwise, $c$ will be carried over in parallel round $i$, and we must show it is also carried over in asynchronous round $i$.

Suppose to the contrary that particle $P_j$ consumes $c$ in parallel round $i$ but not in asynchronous round $i$.
Then $c$ must be a decrement token, and whenever $P_j$ was activated in asynchronous round $i$, it must have been that $P_{j+1}.\bit = 1$ and $P_{j+1}.\tokens$ contained a decrement token $c'$, blocking the consumption of $c$.
By the induction hypothesis, we have that $I_{C_{i-1}}(c') \leq I_{C_{i-1}^A}(c') = j+1$, and since the order of tokens is maintained, we have that $j = I_{C_{i-1}}(c) < I_{C_{i-1}}(c')$.
Combining these expressions, we have $I_{C_{i-1}}(c') = j+1$; i.e., $P_{j+1}$ holds $c'$ just before parallel round $i$.
We will show this situation is impossible: it cannot occur in the parallel execution that $P_j$ is holding a decrement token $c$ it will consume while $P_{j+1}$ is also holding a decrement token $c'$ in the same round.
For $c'$ to have reached $P_{j+1}$, it must have been carried over from $P_j$ in a previous round when $P_j.\bit = 0$.
Since the parallel counter schedule is greedy, the only way $c'$ is still at $P_{j+1}$ in parallel round $i$ is if this carry over occurred in the preceding round, $i-1$.
This carry over would have left $P_j.\bit = 0$ in parallel round $i$, but for $P_j$ to be able to consume $c$ in round $i$, as supposed, we must have that $P_j.\bit = 1$, a contradiction.

Now suppose to the contrary that $c$ is carried over from $P_j$ to $P_{j+1}$ in parallel round $i$ but not in asynchronous round $i$.
Then whenever $P_j$ was last activated in asynchronous round $i$, $P_{j+1}$ must have been holding two counter tokens, say $c'$ and $c''$, where $c'$ is buffered and $c''$ is the token $P_{j+1}$ is currently processing.
Thus, since counter tokens cannot overtake one another (i.e., their order is maintained), $P_{j+1}$ must have been holding $c'$ and $c''$ before asynchronous round $i$ began, i.e., $I_{C_{i-1}^A}(c') = I_{C_{i-1}^A}(c'') = j+1$.
But particles in the parallel setting cannot hold two tokens at once, and since the order of the tokens is maintained, we must have $I_{C_{i-1}}(c'') > I_{C_{i-1}}(c') \geq I_{C_{i-1}}(c) + 1 = j + 1$.
Combining these expressions, we have $I_{C_{i-1}}(c'') > I_{C_{i-1}}(c') \geq j + 1 = I_{C_{i-1}^A}(c'')$, contradicting $C_{i-1}^A \succeq C_{i-1}$.

\medskip

Therefore, $I_{C_i}(c) \leq I_{C_i^A}(c)$ in both cases, and since the choice of $c$ was arbitrary we conclude that $C_i^A \succeq C_i$.
\end{proof}
\fi

So it suffices to bound the number of rounds a greedy parallel counter schedule requires to process its counter operations.
The following lemma shows that the counter can always process a new increment or decrement operation at the start of a parallel round.

\begin{lemma} \label{lem:counterprogress}
Consider any counter token $c$ in any configuration $C_i$ of a greedy parallel counter schedule $(S, (C_0, \ldots, C_t))$. In $C_{i+1}$, $c$ either has been carried over once ($I_{C_{i+1}}(c) = I_{C_i}(c) + 1$) or has been consumed ($I_{C_{i+1}}(c) = \infty$).
\end{lemma}
\ifconf\else
\begin{proof}
This follows directly from Definition~\ref{def:counterschedule}.
If counter token $c$ is held by the unique particle $P$ that will consume it in configuration $C_i$, then by Actions 2 or 3 (if $c$ is an increment or decrement token, respectively), nothing prohibits $P$ from consuming $c$ in parallel round $i+1$.
Since the parallel counter schedule is greedy, this must occur, so $I_{C_{i+1}}(c) = \infty$.

Otherwise, $c$ needs to be carried over from, say, $P_j$ to $P_{j+1}$ where $j = I_{C_i}(c)$.
In the parallel setting, each particle can only store one token at a time.
So the only reason $c$ would not be carried over to $P_{j+1}$ in parallel round $i+1$ is if $P_{j+1}$ was also holding a counter token that needed to but couldn't be carried over in parallel round $i+1$.
But this is impossible, since tokens can always be carried over past the end of the counter, and thus all tokens can be carried over in parallel.
So $I_{C_{i+1}}(c) = I_{C_i}(c) + 1$.
\end{proof}
\fi

Unlike in the asynchronous setting, zero-testing is always available in the parallel setting.

\begin{lemma} \label{lem:parallel0testavail}
The zero-test operation is available at every configuration of a greedy parallel counter schedule.
\end{lemma}
\ifconf\else
\begin{proof}
Recall that zero-testing is unavailable whenever $P_1.\bit = 1$ and $P_1.\tokens = [c^-]$.
This issue stems from ambiguity about where the most significant bit is in the asynchronous setting, since it is possible for an adversarial activation sequence to flood the counter with decrements while temporarily stalling the particle holding the final token $f$.
This results in a configuration where the counter's value is effectively $0$ (with many decrements waiting to be processed), but the counter has not yet shrunk appropriately, bringing $f$ to particle $P_1$.

This is not a concern of the parallel setting; by Lemma~\ref{lem:counterprogress}, we have that each counter token is either carried over or consumed in the next parallel round.
So if $P_1$ is holding a decrement token $c^-$ and $P_1.\bit = 1$, it must be because $P_0 = \ell$ just generated that $c^-$ and forwarded it to $P_1$ in the previous parallel round.
Thus, a conclusive zero-test can be performed at the end of each parallel round.
\end{proof}
\fi

We can synthesize these results to bound the running time of our distributed counter.

\begin{theorem} \label{thm:counterruntime}
Given any nonnegative sequence $S$ of $m$ operations and any fair asynchronous activation sequence $A$, the distributed binary counter processes all operations in $\bigO{m}$ asynchronous rounds.
\end{theorem}
\begin{proof}
Let $(S, (C_0, \ldots, C_t))$ be the greedy parallel counter schedule corresponding to the asynchronous counter schedule defined by $A$ and $S$ in Lemma~\ref{lem:counterdominance}.
By Lemma~\ref{lem:counterprogress}, the leader $\ell$ can generate one new operation from $S$ in every parallel round.
Since we have $m$ such operations, the corresponding parallel execution requires $m$ parallel rounds to generate all operations in $S$.
Also by Lemma~\ref{lem:counterprogress}, assuming in the worst case that all $m$ operations are increments, the parallel execution requires an additional $\lceil \log_2m \rceil$ parallel rounds to process the last operation.
If ever the counter needed to perform a zero-test, we have by Lemmas~\ref{lem:parallel0testavail} and~\ref{lem:0testreliable} that this can be done immediately and reliably.
So all together, processing all operations in $S$ requires $\mathcal{O}(m + \log_2m) = \mathcal{O}(m)$ parallel rounds in the worst case, which by Lemma~\ref{lem:counterdominance} is also an upper bound on the worst case number of asynchronous rounds.
\end{proof}

\section{The Convex Hull Algorithm} \label{sec:multialg}

We now show how a system of $n$ particles each with only constant-size memory can emulate the single-particle algorithm of Section~\ref{sec:soloalg}.
Recall that we assume there are sufficient particles to maintain the binary counters and that the system contains a unique leader particle $\ell$ initially adjacent to the object.
This leader $\ell$ is primarily responsible for emulating the particle with unbounded memory in the single-particle algorithm.
To do so, it organizes the other particles in the system as distributed memory, updating its distances $d_h$ to half-plane $h$ as it moves along the object's boundary.
This is our algorithm's \emph{learning phase}.
In the \emph{formation phase}, $\ell$ uses these complete measurements to lead the other particles in forming the convex hull.
There is no synchronization among the various (sub)phases of our algorithm; for example, some particles may still be finishing the learning phase after the leader has begun the formation phase.

\subsection{Learning the Convex Hull} \label{subsec:learningalg}

The \emph{learning phase} combines the movement rules of the single-particle algorithm (Section~\ref{sec:soloalg}) with the distributed binary counters (Section~\ref{sec:counter}) to enable the leader to measure the convex hull $H(O)$.
Accompanying pseudocode can be found in Appendix~\ref{subapp:learncode}.
We note that there are some nuances in adapting the general-purpose binary counters for use in our convex hull formation algorithm.
For clarity, we will return to these issues in Section~\ref{subsec:counteradapt} after describing this phase.

In the learning phase, each particle $P$ can be in one of three states, denoted $P.\mystate$: \emph{leader}, \emph{follower}, or \emph{idle}.
All non-leader particles are assumed to be initially idle and contracted.
To coordinate the system's movement, the leader $\ell$ orients the particle system as a spanning tree rooted at itself.
This is achieved using the \emph{spanning tree primitive} (see, e.g.,~\cite{Daymude2019}).
If an idle particle $P$ is activated and has a non-idle neighbor, then $P$ becomes a follower and sets $P.\parent$ to this neighbor.
This primitive continues until all idle particles become followers.

Imitating the single-particle algorithm of Section~\ref{sec:soloalg}, $\ell$ performs a clockwise traversal of the boundary of the object $O$ using the right-hand rule, updating its distance counters along the way.
It terminates once it has visited all six half-planes without pushing any of them, which it detects using its terminating bits $b_h$.
In this multi-particle setting, we need to carefully consider both how $\ell$ updates its counters and how it interacts with its followers as it moves.

\paragraph*{Rules for Leader Computation and Movement}

If $\ell$ is expanded and it has a contracted follower child $P$ in the spanning tree that is keeping counter bits, $\ell$ pulls $P$ in a handover.

Otherwise, suppose $\ell$ is contracted.
If all its terminating bits $b_h$ are equal to $1$, then $\ell$ has learned the convex hull, completing this phase.
Otherwise, it must continue its traversal of the object's boundary.
If the zero-test operation is unavailable or if it is holding increment/decrement tokens for any of its $d_h$ counters, it will not be able to move.
Otherwise, let $i \in [6]$ be its next move direction according to the right-hand rule, and let $v$ be the node in direction $i$.
There are two cases: either $v$ is unoccupied, or $\ell$ is blocked by another particle occupying $v$.

In the case $\ell$ is blocked by a contracted particle $P$, $\ell$ can \emph{role-swap} with $P$, exchanging its memory with the memory of $P$.
In particular, $\ell$ gives $P$ its counter bits, its counter tokens, and its terminating bits; promotes $P$ to become the new leader by setting $P.\mystate \gets$ \emph{leader} and clearing $P.\parent$; and demotes itself by setting $\ell.\mystate \gets$ \emph{follower} and $\ell.\parent \gets P$.
This effectively advances the leader's position one node further along the object's boundary.

If either $v$ is unoccupied or $\ell$ can perform a role-swap with the particle blocking it, $\ell$ first calculates whether the resulting move would push one or more half-planes using update vector $\delta_i$.
Let $\mathcal{H}' = \{h \in \mathcal{H} : \delta_{i,h} = -1 \text{ and } d_h = 0\}$ be the set of half-planes being pushed, and recall that since zero-testing is currently available, $\ell$ can locally check if $d_h = 0$.
It then generates the appropriate increment and decrement tokens according to $\delta_i$.
Next, it updates its terminating bits: if it is about to push a half-plane (i.e., $\mathcal{H}' \neq \emptyset$), then it sets $b_h \gets 0$ for all $h \in \mathcal{H}$; otherwise, it can again use zero-testing to set $b_h \gets 1$ for all $h \in \mathcal{H}$ such that $d_h + \delta_{i,h} = 0$.
Finally, $\ell$ performs its move: if $v$ is unoccupied, $\ell$ expands into $v$; otherwise, $\ell$ performs a role-swap with the contracted particle blocking it.

\paragraph*{Rules for Follower Movement}

Consider any follower $P$.
If $P$ is expanded and has no children in the spanning tree nor any idle neighbor, it simply contracts.
If $P$ is contracted and is following the tail of its expanded parent $Q = P.\parent$, it is possible for $P$ to push $Q$ in a handover.
Similarly, if $Q$ is expanded and has a contracted child $P$, it is possible for $Q$ to pull $P$ in a handover.
However, if $P$ is not emulating counter bits but $Q$ is, then it is possible that a handover between $P$ and $Q$ could disconnect the counters (see \figtext~\ref{fig:counterconnect}).
So we only allow these handovers if either $(i)$ both keep counter bits, like $P_3$ and $P_4$ in \figtext~\ref{fig:counterconnect}; $(ii)$ neither keep counter bits, like $Q_2$ and $Q_3$ in \figtext~\ref{fig:counterconnect}; or $(iii)$ one does not keep counter bits while the other holds the final token, like $P_6$ and $R_1$ in \figtext~\ref{fig:counterconnect}.

\begin{figure}[t]
\centering
\includegraphics[scale=0.8]{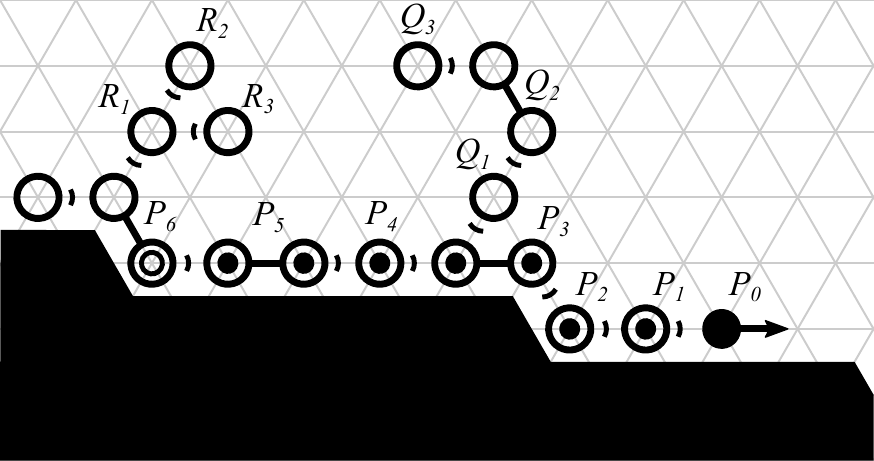}
\caption{The leader $P_0$ (black dot) and its followers (black circles). Followers with dots keep counter bits, and $P_6$ holds the final token. Allowing $Q_1$ to handover with $P_3$ would disconnect the counter, while all other potential handovers are safe.}
\label{fig:counterconnect}
\end{figure}

\subsection{Adapting the Binary Counters for Convex Hull Formation} \label{subsec:counteradapt}

Both the learning phase (Section~\ref{subsec:learningalg}) and the formation phase (Section~\ref{subsec:formationalg}) use the six distance counters $d_h$, for $h \in \mathcal{H}$.
As alluded to in the previous section, we now describe how to adapt the general-purpose binary counters described in Section~\ref{sec:counter} for convex hull formation.
Accompanying pseudocode can be found in Appendix~\ref{subapp:hullcountercode}.

First, since the particle system is organized as a spanning tree instead of a simple path, a particle $P$ must unambiguously decide which neighboring particle keeps the next most significant bit.
Particle $P$ first prefers a child in the spanning tree already holding bits of a counter.
If none exist, a child ``hull'', ``marker'', or ``pre-marker'' particle (see Section~\ref{subsec:formationalg}) is used.
Finally, if none exist, a child on the object's boundary is chosen.
(We prove that at least one of these cases is satisfied in Lemma~\ref{lem:counterextend}).

Second, each particle may participate in up to six $d_h$ counters instead of just one.
Since the different counters never interact with one another, this modification is easily handled by indexing the counter variables by the counter they belong to.
For each half-plane $h \in \mathcal{H}$, the final token $f_h$ denotes the end of the counter $d_h$, increment and decrement tokens are tagged $c_h^+$ and $c_h^-$, respectively, and a particle $P$ keeps bits $P.\bit_h$ and holds tokens $P.\tokens_h$.

Third, the particle system is moving instead of remaining static, which affects the binary counters in two ways.
As described in Section~\ref{subsec:learningalg}, certain handovers must be prohibited to protect the connectivity of the counters.
Role-swaps would also disconnect the counters, since the leader transfers its counter information (bits, tokens, etc.) into the memory of the particle blocking it.
To circumvent this issue, we allow each particle to keep up to two bits of each counter instead of one.
Then, during a role-swap, the leader only transfers its less significant bits/tokens for each counter $d_h$, retaining the information related to the more significant bits and thus keeping the counters connected.

The fourth and final modification to the binary counters is called \emph{bit forwarding}.
As described above, both particles involved in the role-swap are left keeping only one bit instead of two.
Thus, if ever a particle $P$ only has one bit of a counter $d_h$ while the particle $Q$ keeping the next most significant bit(s) has two, $P$ can take the less significant bit and tokens from $Q$.
This ensures that all particles eventually hold two bits again.

Other than these four adaptations, the mechanics of the counter operations remain exactly as in Section~\ref{sec:counter}.
These adaptations increase the memory load per particle by only a constant factor (i.e., by one additional bit per half-plane), so the constant-size memory constraint remains satisfied.
Details of how these adaptations are implemented can be found in Appendix~\ref{subapp:hullcountercode}.

\subsection{Forming the Convex Hull} \label{subsec:formationalg}

The \emph{formation phase} brings as many particles as possible into the nodes of the convex hull $H(O)$.
It is divided into two subphases.
In the \emph{hull closing} subphase, the leader particle $\ell$ uses its binary counters to lead the rest of the particle system along a clockwise traversal of $H(O)$.
If $\ell$ completes its traversal, leaving every node of the convex hull occupied by (possibly expanded) particles, the \emph{hull filling} subphase fills the convex hull with as many contracted particles as possible.

\subsubsection{The Hull Closing Subphase} \label{subsubsec:hullclosing}

When the learning phase ends, the leader particle $\ell$ occupies a position $s \in H(O)$ (by Lemma~\ref{lem:termination}) and its distributed binary counters contain accurate distances to each of the six half-planes $h \in \mathcal{H}$.
The leader's main role during the hull closing subphase is to perform a clockwise traversal of $H(O)$, leading the rest of the particle system into the convex hull.
In particular, $\ell$ uses its binary counters to detect when it reaches one of the six vertices of $H(O)$, at which point it turns $60^{\circ}$ clockwise to follow the next half-plane, and so on.

The particle system tracks the position $s$ that $\ell$ started its traversal from by ensuring a unique \emph{marker} particle occupies it.
The marker particle is prohibited from contracting out of $s$ except as part of a handover, at which point the marker role is transferred so that the marker particle always occupies $s$.
Thus, when $\ell$ encounters the marker particle occupying the next node of the convex hull, it can locally determine that it has completed its traversal and this subphase.

However, there may not be enough particles to close the hull.
Let $n$ be the number of particles in the system and $H = |H(O)|$ be the number of nodes in the convex hull.
If $n < \lceil H / 2\rceil$, eventually all particles enter the convex hull and follow the leader as far as possible without disconnecting from the marker particle, which is prohibited from moving from position $s$.
With every hull particle expanded and unable to move any farther, a token passing scheme is used to inform the leader that there are insufficient particles for closing the hull and advancing to the next subphase.
Upon receiving this message, the leader terminates, with the rest of the particles following suit.

In the following, we give a detailed implementation of this subphase from the perspective of an individual particle $P$.
Accompanying pseudocode can be found in Appendix~\ref{subapp:hullclosecode}.

\paragraph*{Rules for Leader Computation and Movement}

If the leader $\ell$ is holding the ``all expanded'' token and does not have the marker particle in its neighborhood --- indicating that there are insufficient particles to complete this subphase --- it generates a ``termination'' token and passes it to its child in the spanning tree.
It then terminates by setting $\ell.\mystate \gets$ \emph{finished}.

Otherwise, if $\ell$ is expanded, there are two cases.
If $\ell$ has a contracted hull child $Q$ (i.e., a child $Q$ with $Q.\mystate =$ \emph{hull}), $\ell$ performs a pull handover with $Q$.
If $\ell$ does not have any hull children but does have a contracted follower child $Q$ keeping counter bits, then this is its first expansion of the hull closing subphase and the marker should occupy its current tail position.
So $\ell$ sets $Q.\mystate \gets$ \emph{pre-marker} and performs a pull handover with $Q$ (see \figtext~\ref{fig:marker:a}--\ref{fig:marker:b}).

\begin{figure}[t]
\centering
\begin{subfigure}{.24\textwidth}
	\centering
	\includegraphics[scale=.75]{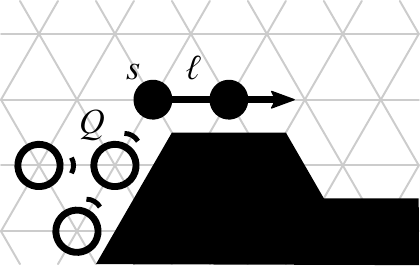}
	\caption{\centering}
	\label{fig:marker:a}
\end{subfigure}
\hfill
\begin{subfigure}{.24\textwidth}
	\centering
	\includegraphics[scale=.75]{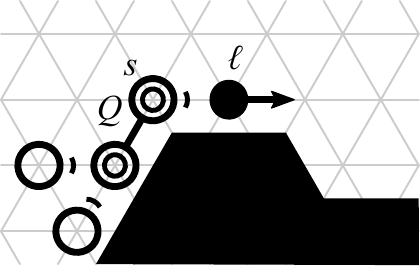}
	\caption{\centering}
	\label{fig:marker:b}
\end{subfigure}
\hfill
\begin{subfigure}{.24\textwidth}
	\centering
	\includegraphics[scale=.75]{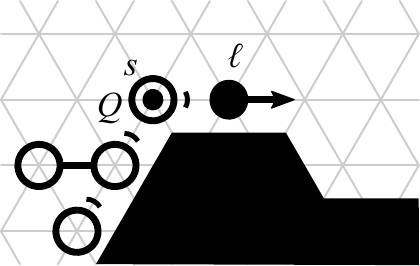}
	\caption{\centering}
	\label{fig:marker:c}
\end{subfigure}
\hfill
\begin{subfigure}{.24\textwidth}
	\centering
	\includegraphics[scale=.75]{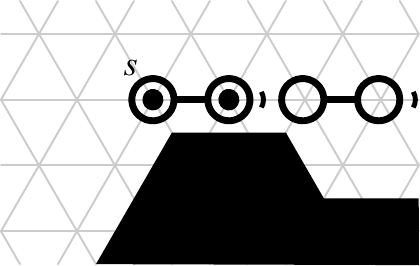}
	\caption{\centering}
	\label{fig:marker:d}
\end{subfigure}
\caption{(a) After expanding for the first time, the leader $\ell$ occupies the starting position $s$ with its tail.
(b) After performing a handover with $\ell$, follower child $Q$ becomes the pre-marker (inner circles).
(c) When $Q$ contracts, it becomes the marker (inner dot).
(d) If there are insufficient particles to close the hull, the marker particle will eventually become expanded and unable to contract without vacating position $s$.}
\label{fig:marker}
\end{figure}

During its hull traversal, $\ell$ keeps a variable $\ell.\plane \in \mathcal{H}$ indicating which half-plane boundary it is currently following.
It checks if it has reached the next half-plane by zero-testing: if the distance to the next half-plane is $0$, $\ell$ updates $\ell.\plane$ accordingly.
It then inspects the next node of its traversal along $\ell.\plane$, say $v$.
If $v$ is occupied by the marker particle $Q$, then $\ell$ has completed the hull closing subphase; it updates $Q.\mystate \gets$ \emph{finished} and then advances to the hull filling subphase (Section~\ref{subsubsec:hullfilling}).
Otherwise, if $\ell$ is contracted, it continues  its traversal of the convex hull by either expanding into node $v$ if $v$ is unoccupied or by role-swapping with the particle blocking it, just as it did in the learning phase.

\paragraph*{Rules for the Marker Particle Logic}

The marker role must be passed between particles so that the marker particle always occupies the position at which the leader started its hull traversal.
Whenever a contracted marker particle $P$ expands in a handover with its parent, it remains a marker particle.
When $P$ subsequently contracts as a part of a handover with a contracted child $Q$, $P$ becomes a hull particle and $Q$ becomes a \emph{pre-marker}.
Finally, when the pre-marker $Q$ contracts --- either on its own or as part of a handover with a contracted child --- $Q$ becomes the marker particle (see \figtext~\ref{fig:marker:c}).

Importantly, the marker particle $P$ never contracts outside of a handover, as this would vacate the leader's starting position (see \figtext~\ref{fig:marker:d}).
If $P$ is ever expanded but has no children or idle neighbors, it generates the ``all expanded'' token and passes it forward along expanded particles only.
If this ultimately causes the leader to learn there are insufficient particles to close the hull (as described above) and the ``termination'' token is passed all the way back to $P$, $P$ terminates by consuming the termination token and becoming finished.

\paragraph*{Rules for Follower and Hull Particle Behavior}

Follower particles move just as they did in the learning phase, with two additional conditions.
First, if ever a follower is involved in a handover with the pre-marker or marker particle, their states are updated as described above.
Second, follower particles never perform handovers with hull particles.

Hull particles are simply follower particles that have joined the convex hull.
They only perform handovers with the leader and other hull particles.
Additionally, they're responsible for passing the ``all expanded'' and ``termination'' tokens: if an expanded hull particle $P$ holds the ``all expanded'' token and its parent is also expanded, $P$ passes this token to its parent.
If a hull particle $P$ is holding the ``termination'' token, it terminates by passing this token to its hull or marker child and becoming finished.

\subsubsection{The Hull Filling Subphase} \label{subsubsec:hullfilling}

The hull filling subphase is the final phase of the algorithm.
It begins when the leader $\ell$ encounters the marker particle in the hull closing subphase, completing its traversal of the hull.
At this point, the hull is entirely filled with particles, though some may be expanded.
The remaining followers are either outside the hull or are trapped between the hull and the object.
The goal of this subphase is to $(i)$ allow trapped particles to escape outside the hull, and $(ii)$ use the followers outside the hull to ``fill in'' behind any expanded hull particles, filling the hull with as many contracted particles as possible.

At a high level, this subphase works as follows.
The leader $\ell$ first becomes finished.
Each hull particle then also becomes finished when its parent is finished.
A finished particle $P$ labels a neighboring follower $Q$ as either \emph{trapped} or \emph{filler} depending on whether $Q$ is inside or outside the hull, respectively.
This can be determined locally using the relative position of $Q$ to the parent of $P$, which is the next particle on the hull in a clockwise direction.
A trapped particle performs a coordinated series of movements with a neighboring finished particle to effectively take its place, ``pushing'' the finished particle outside the hull as a filler particle.
Filler particles perform a clockwise traversal of the surface of the hull (i.e., the finished particles) searching for an expanded finished particle to handover with.
Doing so effectively replaces a single expanded finished particle on the hull with two contracted ones.

There are two ways the hull filling subphase can terminate.
Recall that $n$ is the number of particles in the system and $H = |H(O)|$ is the number of nodes in the convex hull.
If $n \geq H$, the entire hull can be filled with contracted particles.
To detect this event, a token is used that is only passed along contracted finished particles.
If it is passed around the entire hull, termination is broadcast so that all particles (including the extra ones outside the hull) become finished.
However, it may be that $\lceil H/2 \rceil \leq n < H$; that is, there are enough particles to close the hull but not enough to fill it with all contracted particles.
In this case, all particles will still eventually join the hull and become finished.

Detailed pseudocode for this subphase can be found in Appendix~\ref{subapp:hullfillcode}.
In the following, we describe the local rules underlying the three important primitives for this subphase.

\paragraph*{Freeing Trapped Particles}

\begin{figure}[t]
\centering
\begin{subfigure}{.19\textwidth}
	\centering
	\includegraphics[scale=.78]{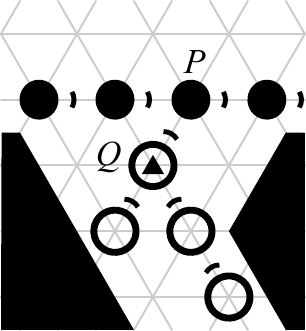}
	\caption{\centering}
	\label{fig:trapped:a}
\end{subfigure}
\hfill
\begin{subfigure}{.19\textwidth}
	\centering
	\includegraphics[scale=.78]{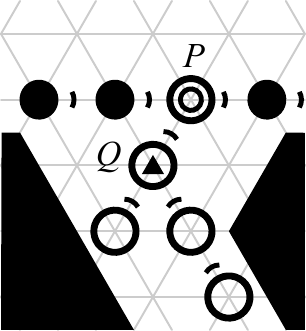}
	\caption{\centering}
	\label{fig:trapped:b}
\end{subfigure}
\hfill
\begin{subfigure}{.19\textwidth}
	\centering
	\includegraphics[scale=.78]{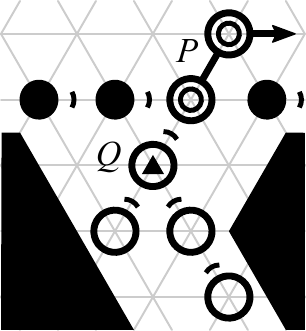}
	\caption{\centering}
	\label{fig:trapped:c}
\end{subfigure}
\hfill
\begin{subfigure}{.19\textwidth}
	\centering
	\includegraphics[scale=.78]{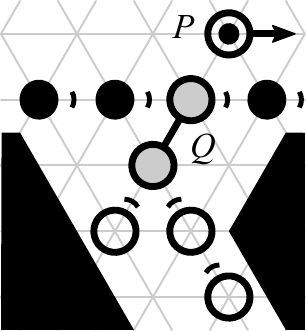}
	\caption{\centering}
	\label{fig:trapped:d}
\end{subfigure}
\hfill
\begin{subfigure}{.19\textwidth}
	\centering
	\includegraphics[scale=.78]{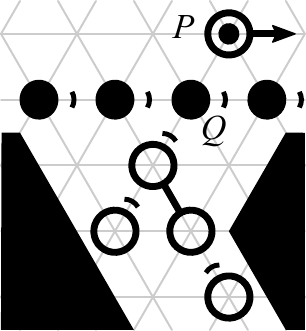}
	\caption{\centering}
	\label{fig:trapped:e}
\end{subfigure}
\caption{Freeing a trapped particle.
(a) A finished particle $P$ marks a neighboring follower $Q$ on the interior of the hull as trapped (inner triangle).
(b) $Q$ marks its parent $P$ as a pre-filler (inner circle).
(c) $P$ expands outside the hull.
(d) In a handover between $P$ and $Q$, $P$ becomes a filler (inner dot) and $Q$ becomes pre-finished (gray).
(e) $Q$ contracts and becomes finished.}
\label{fig:trapped}
\end{figure}

Suppose a finished particle $P$ has labeled a neighboring contracted particle $Q$ as trapped (see \figtext~\ref{fig:trapped:a}).
In doing so, $P$ sets itself as the parent of $Q$.
When $Q$ is next activated, it sets $P.\mystate \gets$ \emph{pre-filler} (see \figtext~\ref{fig:trapped:b}).
This indicates to $P$ that it should expand towards the outside of the hull as soon as possible (\figtext~\ref{fig:trapped:c}).
Once $P$ has expanded, $P$ and $Q$ perform a handover (\figtext~\ref{fig:trapped:d}).
This effectively pushes $P$ out of the hull, where it becomes a filler particle, and expands $Q$ into the hull, where it becomes pre-finished.
Finally, whenever $Q$ contracts --- either on its own or during a handover --- it becomes finished, taking the original position and role of $P$ (\figtext~\ref{fig:trapped:e}).

\paragraph*{Filling the Hull}

\begin{figure}[t]
\centering
\begin{subfigure}{.32\textwidth}
	\centering
	\includegraphics[scale=.85]{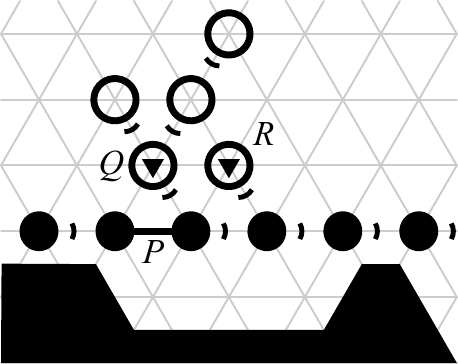}
	\caption{\centering}
	\label{fig:filler:a}
\end{subfigure}
\hfill
\begin{subfigure}{.32\textwidth}
	\centering
	\includegraphics[scale=.85]{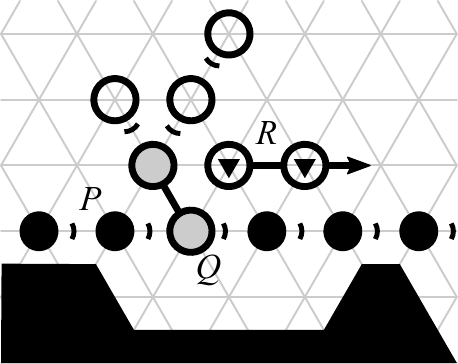}
	\caption{\centering}
	\label{fig:filler:b}
\end{subfigure}
\hfill
\begin{subfigure}{.32\textwidth}
	\centering
	\includegraphics[scale=.85]{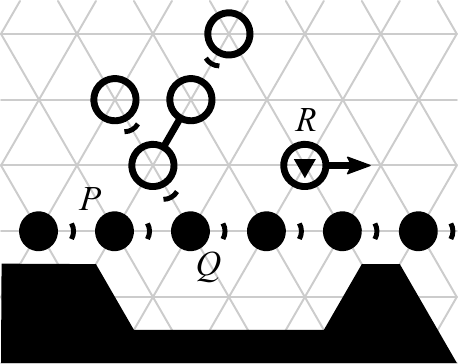}
	\caption{\centering}
	\label{fig:filler:c}
\end{subfigure}
\caption{Some movements of filler particles.
(a) A finished particle $P$ marks neighboring followers $Q$ and $R$ on the exterior of the hull as fillers (inner triangle).
(b) $Q$ performs a handover with $P$ to fill the hull, becoming pre-finished (gray), while $R$ expands along a clockwise traversal of the hull.
(c) $Q$ contracts and becomes finished.}
\label{fig:filler}
\end{figure}

A particle $P$ becomes a filler either by being labeled so by a neighboring finished particle or by being ejected from the hull while freeing a trapped particle, as described above.
If $P$ is expanded, it simply contracts if it has no children or idle neighbors, or performs a pull handover with a contracted follower child if it has one.
If $P$ is contracted, it finds the next node $v$ on its clockwise traversal of the hull.
$P$ simply expands into $v$ unless the first occupied node clockwise from $v$ is occupied by the tail of an expanded finished particle $Q$.
In this case, $P$ performs a push handover with $Q$, sets $Q$ to be its parent, and becomes pre-finished.
Whenever $P$ next contracts --- either on its own or during a handover --- it becomes finished.
An example of a some movements of filler particles can be found in \figtext~\ref{fig:filler}.

\paragraph*{Detecting Termination}

Before $\ell$ finishes at the start of this subphase, it generates an ``all contracted'' token containing a counter $t$ initially set to $0$.
This token is passed backwards along the hull to contracted finished particles only.
Whenever the token is passed through a vertex of the convex hull, the counter $t$ is incremented.
Thus, if a contracted finished particle is ever holding the ``all contracted'' token and its counter $t$ is equal to $7$, it terminates by consuming the ``all contracted'' token and broadcasting ``termination'' tokens to all its neighbors.
Whenever a particle receives a termination token, it also terminates by becoming finished.

\subsection{Correctness} \label{subsec:multicorrect}

\paragraph*{Correctness of the Counters}

We first build on the correctness proofs of Section~\ref{sec:counter} to show that the adapted distributed binary counters described in Section~\ref{subsec:counteradapt} remain correct.
Recall that there are six $d_h$ counters maintained by a spanning tree of follower particles rooted at the leader $\ell$.
Because the $d_h$ counters never interact with one another, we can analyze the correctness of each counter independently.
Also recall that we allow each particle to keep up to two bits of each counter instead of one.
Since the order of the bits is maintained, this does not affect correctness.
We begin by proving several general results.
Throughout this section, recall that $B = |B(O)|$ denotes the length of the object's boundary, and $H = |H(O)|$ denotes the length of the object's convex hull.

\begin{lemma} \label{lem:counterconnect}
The distributed binary counters never disconnect.
\end{lemma}
\begin{proof}
\ifconf
\else
By the spanning forest primitive~\cite{Derakhshandeh2017}, the particle system cannot physically become disconnected.
So the only way to disconnect a counter $d_h$ is to insert a follower that is not keeping bits of $d_h$ between two particles that are.
There are two ways this could occur.
A contracted follower not keeping bits of $d_h$ could perform a handover with an expanded follower that is (as in \figtext~\ref{fig:counterconnect}), separating the counter from its more significant bits.
Alternatively, the leader $\ell$ could role-swap without leaving behind a bit to keep $d_h$ connected.
Both of these movements were explicitly forbidden in Section~\ref{subsec:learningalg}, so the counters remain connected.
\fi
\end{proof}

Next, we prove two useful results regarding the lengths of the distributed binary counters.

\begin{lemma} \label{lem:counterlen}
Let $L$ be the path of nodes traversed by leader $\ell$ from the start of the algorithm to its current position.
Then there are at most $\lfloor \log_2\min\{|L|, H\}\rfloor + 1$ particles holding bits of a distributed binary counter $d_h$.
\end{lemma}
\begin{proof}
It is easy to see that the value of $d_h$ is at most $\min\{|L|, H\}$: $\ell$ cannot be further from its current estimation of half-plane $h$ than the number of moves it has made, and its distance from the true half-plane $h$ is trivially upper bounded by the length of the convex hull.
Since exactly $\lfloor \log_2b\rfloor + 1$ bits are needed to store a binary value $b$, we have that $\lfloor \log_2\min\{|L|, H\}\rfloor + 1$ bits suffice to store $d_h$.
Each particle maintaining $d_h$ holds at least one bit, so there are at most $\lfloor \log_2\min\{|L|, H\}\rfloor + 1$ such particles.
\end{proof}

\begin{lemma} \label{lem:followerlen}
Let $L$ be the path of nodes traversed by leader $\ell$ from the start of the algorithm to its current position.
Then there are at least $\min\{|\mathcal{P}|, \lceil |L| / 2\rceil\}$ particles including $\ell$ along $L$.
\end{lemma}
\begin{proof}
Argue by induction on $|L|$.
If $|L| = 1$, then $\min\{|\mathcal{P}|, \lceil |L| / 2\rceil\} = 1$ and $\ell$ is the only particle on its traversal path.
So consider any $|L| > 1$, and suppose that the lemma holds for all $|L'| < |L|$.
In particular, consider the subpath $L' \subseteq L$ containing all nodes of $L$ except the one $\ell$ most recently moved into; thus, $|L'| = |L| - 1$.
By the induction hypothesis, there were at least $\min\{|\mathcal{P}|, \lceil (|L| - 1) / 2\rceil\}$ particles including $\ell$ on $L'$.
We show that after $\ell$ moves into the $|L|$-th node of its traversal, there are at least $\min\{|\mathcal{P}|, \lceil |L| / 2\rceil\}$ particles along $L$.

If $|\mathcal{P}| \leq \lceil (|L| - 1) / 2\rceil$, then all particles (including $\ell$) were on $L'$.
Regardless of how $\ell$ moves into the $|L|$-th node of its traversal --- i.e., either by an expansion or a role-swap --- it cannot remove a particle as its follower.
So there remain $|\mathcal{P}| \geq \min\{|\mathcal{P}|, \lceil |L| / 2\rceil\}$ particles along $L$.

Otherwise, if $|\mathcal{P}| > \lceil (|L| - 1) / 2\rceil$, there are two cases to consider.
If $|L| - 1$ is odd, then there were at least $|L|/2$ particles on $L'$, a path of $|L| - 1$ nodes.
Thus, at least one particle on $L'$ was contracted.
Via successive handovers, $\ell$ could eventually become contracted and perform its expansion or role-swap into the $|L|$-th node of its traversal, which again could not remove any of its followers.
So there are at least $|L|/2 \geq \min\{|\mathcal{P}|, \lceil |L| / 2\rceil\}$ particles along $L$.

The second case is if $|L| - 1$ is even, implying that there were at least $(|L| - 1)/2$ particles on $L'$, a path of $|L| - 1$ nodes.
If there were strictly more than $(|L| - 1)/2$ particles on $L'$, at least one of them must have been contracted, and an argument similar to the odd case applies here as well.
However, if there were exactly $(|L| - 1)/2$ particles on $L'$, then every particle along $L'$ was expanded, including $\ell$.
Thus, some new follower must have joined $L'$ in order to enable successive handovers that allowed $\ell$ to contract and then move into the $|L|$-th node of its traversal.
So there are $(|L| - 1)/2 + 1 = (|L| + 1)/2 \geq \min\{|\mathcal{P}|, \lceil |L| / 2\rceil\}$ particles along $L$.
\end{proof}

These two lemmas are the key to proving the safety of our algorithm's use of the distributed binary counters.
In particular, we now show that the counters never intersect themselves --- corrupting the order of the bits --- and that there are always enough particles to maintain the counters.

\begin{corollary} \label{cor:counterintersect}
The distributed binary counters never intersect.
\end{corollary}
\begin{proof}
Suppose to the contrary that $\ell$ forms a cycle $\ell = P_1, \ldots, P_k, P_{k+1} = P_1$ in the spanning tree such that every particle $P_i$ on the cycle is keeping bits of a counter $d_h$.
Recall that $\ell$ first traverses $B(O)$ in the learning phase until it accurately measures the convex hull, at which point it traverses $H(O)$ in the hull closing subphase.
The particles maintaining counters only exist on this traversal path.
Thus, any cycle $\ell$ could create has length $k \geq H$.
But by Lemma~\ref{lem:counterlen}, there are at most $\lfloor \log_2\min\{|L|, H\}\rfloor + 1$ particles holding bits of a given counter, and this value is maximized when $|L| \geq H$.
So the cycle must have length at least $H$ but at most $\lfloor \log_2H\rfloor + 1$, which is impossible because $H \geq 6$ due to the geometry of the triangular lattice, a contradiction.
\end{proof}

\begin{corollary} \label{cor:countersufficient}
There are always enough particles to maintain the distributed binary counters.
\end{corollary}
\begin{proof}
We prove that the number of particles holding bits of a given counter never exceeds the number of particles following leader $\ell$ along its traversal path.
By Lemmas~\ref{lem:counterlen} and~\ref{lem:followerlen}, it suffices to show $\lfloor \log_2\min\{m, H\}\rfloor + 1 \leq \min\{|\mathcal{P}|, \lceil m / 2\rceil\}$ for any number of nodes $m \geq 1$ traversed by $\ell$.
Using the assumption that $|\mathcal{P}| > \log_2H$, careful case analysis shows that this inequality holds.
\end{proof}

The following lemma shows that each particle can unambiguously decide which particle holds the next most significant bit of a counter when the particle system is structured as a spanning tree instead of a simple path.

\begin{lemma} \label{lem:counterextend}
Suppose a distributed binary counter $d_h$ is maintained by particles $\ell = P_1, \ldots, P_k$, where $k \leq \lfloor \log_2H\rfloor$.
Then for every $i \in \{1, \ldots, k\}$, $P_i$ can identify the particle responsible for the next most significant bit of $d_h$ unambiguously.
\end{lemma}
\begin{proof}
Recall from Section~\ref{subsec:counteradapt} that $P_i$ identifies the particle responsible for the next most significant bit of $d_h$ by preferring, in this order, a child already holding counter bits, a child hull or (pre-)marker particle, or a child on $B(O)$.
We show such a particle exists and is unambiguous by induction on~$k$.

If $k = 1$, then $\ell = P_1$ is the only particle keeping bits of $d_h$ and thus has no children keeping counter bits.
If $\ell$ is only holding one bit of $d_h$, then $\ell$ itself could hold the next most significant bit.
So suppose $\ell$ is holding two bits of $d_h$, implying that $\ell$ has expanded or role-swapped at least twice.
In the learning phase, no hull or (pre-)marker particles exist.
Since $\ell$ only traverses $B(O)$ in this phase, it always has a follower child on $B(O)$.
In the hull closing subphase, $\ell$ only traverses $H(O)$, and all particles on $H(O)$ are either hull particles or the (pre-)marker particle.
The hull filling subphase does not use counters.
Thus, in all phases, $\ell$ can unambiguously identify the particle responsible for the next most significant bit.

Now consider any $1 < k \leq \lfloor \log_2H\rfloor$, and suppose the lemma holds for all $k' < k$.
For all $1 \leq i < k$, $P_{i+1}$ is the unambiguous child of $P_i$ already holding bits of $d_h$.
So consider $P_k$.
If $P_k$ is only holding one bit of $d_h$, then $P_k$ itself could hold the next most significant bit.
So suppose $P_k$ is holding two bits of $d_h$.
If $P_k$ is a hull particle, it has exactly one child also on the convex hull, and this child must be a hull particle or the (pre-)marker particle.
Otherwise (i.e., if $P_k$ is not a hull particle), we know by the induction hypothesis that $P_k$ is either the (pre-)marker particle or a follower on $B(O)$.
In order for $P_k$ to be holding two bits of $d_h$, the value of $d_h$ must be at least $2^k$ since $d_h$ is connected by Lemma~\ref{lem:counterconnect}.
This implies $\ell$ has expanded or role-swapped at least $2^k$ times, so by Lemma~\ref{lem:followerlen} there are at least $\min\{|\mathcal{P}|, \lceil (2^k + 1) / 2\rceil\}$ particles following $\ell$ along its traversal path.
To identify a unique child follower of $P_k$ on $B(O)$, it suffices to show $\min\{|\mathcal{P}|, \lceil (2^k + 1) / 2\rceil\} \geq k$, i.e., that there are more followers extending along $B(O)$ than are currently holding bits of $d_h$.
By our assumption that $|\mathcal{P}| > \log_2H$ and our supposition that $k \leq \lfloor \log_2H\rfloor$, we have:
\[|\mathcal{P}| > \log_2H \geq \lfloor \log_2H\rfloor + 1 > k.\]
Since $2^k + 1$ is always odd whenever $k > 1$, we have $\lceil (2^k + 1) / 2\rceil = 2^{k-1} + 1$, which is strictly greater than $k$ for all $k > 1$.
\end{proof}

Thus, the counters are all extended along the same, unambiguous path of particles.
To conclude our results on the distributed binary counters, we show that bit forwarding moves the bits of all six counters towards the leader as far as possible.

\begin{lemma} \label{lem:counterbitforward}
If $\ell$ only has one bit of a distributed binary counter $d_h$ and is not holding the final token $f_h$ at time $t$, then there exists a time $t' > t$ when $\ell$ either has two bits of $d_h$ or is holding $f_h$.
\end{lemma}
\begin{proof}
\ifconf
\else
Suppose $\ell$ is only emulating one bit of a counter $d_h$ and is not holding $f_h$ at time $t$.
Argue by induction on $i$, the number of consecutive particles starting at $\ell = P_1$ that are only emulating one bit of $d_h$ and are not holding $f_h$.
If $i = 1$, then $P_2$ must either be $(i)$ emulating two bits of $d_h$, $(ii)$ emulating the most significant bit (MSB) of $d_h$ and holding $f_h$, or $(iii)$ only holding $f_h$.
In cases $(i)$ and $(ii)$, $\ell$ can take the less significant bit from $P_2$ during its next activation (say, at time $t' > t$) while in case $(iii)$ $\ell$ can take $f_h$ instead.

Now suppose $i > 1$ and the induction hypothesis holds up to $i - 1$.
Then $P_{i-1}$ is only emulating one bit of $d_h$ and is not holding $f_h$ while $P_i$ satisfies one of the three cases above.
As in the base case, after the next activation of $P_{i-1}$ (say, at $t_1 > t$), $P_{i-1}$ is either emulating two bits of $d_h$ or is holding $f_h$.
Therefore, by the induction hypothesis, there exists a time $t' > t_1$ when $\ell$ is also either emulating two bits of $d_h$ or holding $f_h$.
\fi
\end{proof}

\paragraph*{Correctness of the Learning Phase}

To prove the learning phase is correct, we must show that the leader $\ell$ obtains an accurate measurement of the convex hull by moving and performing zero-tests, emulating the single particle algorithm of Section~\ref{sec:soloalg}.
We already proved in Lemmas~\ref{lem:0testavailable} and~\ref{lem:0testreliable} that $\ell$ will always eventually be able to perform a reliable zero-test.
So we now prove the correctness of the particle system's movements.
This relies in part on previous work on the spanning forest primitive~\cite{Derakhshandeh2017}, where movement for a spanning tree following a leader particle was shown to be correct.
In fact, the correctness of our algorithm's follower movements follows directly from this previous analysis, so it remains to show the leader's movements are correct.

\begin{lemma} \label{lem:ellmove}
If $\ell$ is contracted, it can always eventually expand or role-swap along its clockwise traversal of $B(O)$.
If $\ell$ is expanded, it can always eventually perform a handover with a follower.
\end{lemma}
\begin{proof}
First suppose $\ell$ is contracted.
Leader $\ell$ can only move if its zero-test operation is available for all of its $d_h$ counters, which must eventually be the case by Lemma~\ref{lem:0testavailable}.
Let $v$ be the next clockwise node on $B(O)$.
If $v$ is unoccupied, $\ell$ can simply expand into node $v$.
Otherwise, $\ell$ needs to perform a role-swap with the particle occupying $v$.
This is only allowed when, for each counter $d_h$, $\ell$ holds two bits or the final token $f_h$.
Lemma~\ref{lem:counterbitforward} shows this is always eventually true, implying $\ell$ can perform the role-swap.
In either case, $\ell$ moves into $v$.

Now suppose $\ell$ is expanded.
By previous work on the spanning forest primitive~\cite{Derakhshandeh2017}, some follower child $P$ of $\ell$ will eventually contract.
Thus, $\ell$ can perform a pull handover with $P$ in its next activation to become contracted.
\end{proof}

By Lemma~\ref{lem:ellmove}, we have that the leader $\ell$ can exactly emulate the movements of the single particle in Section~\ref{sec:soloalg}.
Thus, as a direct result of Theorem~\ref{thm:singlecorrect}, $\ell$ completes the learning phase with an accurate measurement of the convex hull of $O$.

\paragraph*{Correctness of the Hull Formation Phase}

The hull formation phase begins with the leader $\ell$ occupying its ``starting position'' $s \in H(O) \cap B(O)$.
Recall that in the hull closing subphase, $\ell$ uses its binary counters to perform a clockwise traversal of the convex hull $H(O)$, leading the rest of the particle system into the convex hull.
The particle system tracks the starting position $s$ by ensuring a marker particle always occupies it, as we now prove.

\begin{lemma} \label{lem:marker}
The starting position $s$ is always occupied by the leader or (pre-)marker particle.
\end{lemma}
\begin{proof}
Initially, the leader $\ell$ occupies $s$.
When it expands into the first node of $H(O)$, its tail still occupies $s$.
When it contracts out of $s$ as part of a handover with a contracted follower child $P$, it sets $P$ as the pre-marker particle; at this point, the head of $P$ occupies $s$.
Whenever a pre-marker particle contracts to occupy $s$ only --- either on its own or as part of a handover with a contracted child --- it becomes the marker particle.
A marker particle $P$ may expand so that its tail still occupies $s$, but can only contract out of $s$ as part of a handover with a contracted child, which $P$ then sets as the pre-marker particle.
Thus, in all cases, $s$ is either occupied by the leader, the pre-marker, or the marker particle.
\end{proof}

If there are insufficient particles to close the hull, we must show that the particle system fills as much of the hull as possible and then terminates.

\begin{lemma} \label{lem:insufficienthullclose}
If there are fewer than $\lceil H/2 \rceil$ particles in the system, each particle will eventually terminate, expanded over two nodes of $H(O)$.
\end{lemma}
\begin{proof}
By nearly the same argument as for Lemma~\ref{lem:ellmove}, $\ell$ will always eventually move along its traversal of $H(O)$, guided by its counters that continuously update the distances to each half-plane.
However, by Lemma~\ref{lem:marker}, the starting position $s$ cannot be vacated by the marker particle unless another particle replaces it in a handover.
Thus, $\ell$ will be able to traverse at most $2 \cdot |\mathcal{P}|$ nodes of $H(O)$ before all particles in the system are expanded, unable to move any further.
By supposition, $|\mathcal{P}| < \lceil H/2\rceil$: if $H$ is even, then $\lceil H/2\rceil = H/2$ and thus $2\cdot |\mathcal{P}| \leq H - 1$; if $H$ is odd, then $\lceil H/2\rceil = (H+1)/2$ and thus $2\cdot |\mathcal{P}| \leq 2((H+1)/2 - 1) = H - 1$.
Thus, there are insufficient particles to close the hull, even if all particles expand.

When the marker particle is expanded and has no children, which must occur by the above argument, it generates the ``all expanded'' token $all_{exp}$.
Because the $all_{exp}$ token is only passed towards the leader by expanded particles, we are guaranteed that every particle from the marker up to the particle currently holding $all_{exp}$ is expanded.
Thus, if $\ell$ ever receives the $all_{exp}$ token but does not have the marker particle in its neighborhood, $\ell$ can locally decide that there are insufficient particles to close the hull.
Termination is then broadcast from $\ell$.
\end{proof}

Assuming there are sufficient particles to close the hull, we must show that the leader successfully completes its traversal of $H(O)$ and advances to the hull filling subphase.

\begin{lemma} \label{lem:hullclosing}
If there are at least $\lceil H/2 \rceil$ particles in the system, then the leader $\ell$ will complete its traversal of $H(O)$, closing the hull.
\end{lemma}
\begin{proof}
Once again, by nearly the same argument as for Lemma~\ref{lem:ellmove}, $\ell$ will always eventually move along its traversal of $H(O)$.
As in Lemma~\ref{lem:insufficienthullclose}, $\ell$ will be able to traverse at most $2 \cdot |\mathcal{P}|$ nodes of $H(O)$.
By supposition, since $|\mathcal{P}| \geq \lceil H/2\rceil$, we have that $2 \cdot |\mathcal{P}| \geq H$.
Thus, there are enough particles for $\ell$ to close the hull.

So it remains to show that the ``all expanded'' token $all_{exp}$ does not cause $\ell$ to terminate incorrectly when there are sufficient particles to close the hull.
By Lemma~\ref{lem:ellmove}, $\ell$ has completed at least one traversal of $B(O)$.
Combining Lemma~\ref{lem:followerlen} with our supposition that $|\mathcal{P}| \geq \lceil H/2\rceil$ and the fact that $B \geq H$, we have that $\ell$ has at least
\[\min\left\{|\mathcal{P}|, \left\lceil \frac{B}{2}\right\rceil\right\} - 1 \geq \min\left\{\left\lceil \frac{H}{2}\right\rceil, \left\lceil \frac{B}{2}\right\rceil\right\} - 1 \geq \left\lceil \frac{H}{2}\right\rceil - 1\]
particles following it.
Thus, in order for the marker particle to be expanded and have no children --- allowing it to generate the $all_{exp}$ token --- there must be at least $\lceil H / 2\rceil$ particles from the marker particle to $\ell$ all on the hull.
If the $all_{exp}$ token is eventually passed to $\ell$, then all of these particles from the marker to $\ell$ must be expanded.
So there must be exactly $H / 2$ of them, since they are all expanded but must fit in the $H$ nodes of the convex hull.
Therefore, either $\ell$ receives the $all_{exp}$ token but has already closed the hull or $\ell$ never receives the $all_{exp}$ token.
In either case, $\ell$ closes the hull and can advance to the hull filling subphase.
\end{proof}

The hull filling subphase begins when $\ell$ encounters the marker particle, closing the hull and finishing.
At this point, the hull may be occupied by both expanded and contracted particles.
We must show that this subphase fills the hull with as many contracted particles as possible.

\begin{lemma} \label{lem:hullfillprogress}
If $H(O)$ is closed but not filled with all contracted particles and there exists a particle occupying a node not in $H(O)$, then at least one particle can make progress towards filling another hull position with a contracted particle.
\end{lemma}
\begin{proof}
The main idea of this argument is to categorize all types of particles that occupy nodes outside $H(O)$ and then order these categories such that if no particles in the first $i$ categories exist, then a particle in the $(i+1)$-th category must be able to make progress.

The first category contains all types of particles that are able to make progress without needing changes in their neighborhoods.
Any idle particle adjacent to a non-idle particle can become a follower in its next activation.
Similarly, any hull particle adjacent to a finished particle can become finished in its next activation.

If no particles from the first category exist to make progress independently, we show a particle from this second category can make progress.
Any expanded particle with no children can contract in its next activation, since there are no idle particles adjacent to non-idle particles.
Since no hull particles are adjacent to finished particles, all hull particles must be finished.
Thus, any contracted follower adjacent to a node of $H(O)$ will be labeled as either trapped or filler by a neighboring finished particle in its next activation.
Moreover, any trapped particle must have a finished parent: if this parent is expanded, the trapped particle can perform a handover with it to become pre-finished; otherwise, the trapped particle can mark this parent as a pre-filler.

Now consider a third category, assuming no particles from the first two categories exist.
Among expanded followers waiting to contract in a handover, at least one expanded follower must have a contracted follower child to handover with because all followers with no children are contracted.
Any expanded pre-filler must have a contracted trapped particle it is freeing, so these particles can perform a handover in their next activation, causing the trapped particle to become pre-finished and the pre-filler to become a filler.

The fourth category follows from the third.
Any expanded filler or pre-finished particle must be waiting to perform a handover with a contracted follower child since all expanded particles with no children have already contracted.
But all followers are now contracted and not adjacent to nodes of $H(O)$.
Thus, any expanded filler or pre-finished particle waiting to perform a handover with a contracted follower child can do so in its next activation.

The fifth category follows from the third and fourth.
From the fourth category, we can now assume all filler particles are contracted.
So any contracted filler particles that can handover with a neighboring expanded finished particle do so, becoming pre-finished.
But from the third category, we know there are no expanded pre-fillers protruding onto the surface of the convex hull.
So there must exist a contracted filler whose next node on its clockwise traversal of the surface of $H(O)$ is unoccupied, and this contracted filler can expand in its next activation.

The final category contains contracted pre-fillers needing to expand outwards, onto the exterior of the convex hull.
From the previous categories, we can assume that there are no longer any followers or fillers on the surface of the convex hull.
Thus, nothing is blocking a contracted pre-filler from expanding outwards in its next activation.
Therefore, as these six categories are exhaustive, we conclude that as long as there exists a particle occupying a node outside $H(O)$, at least one particle can make progress.
\end{proof}

Applying Lemma~\ref{lem:hullfillprogress} iteratively, we can immediately conclude that the convex hull is eventually filled with all contracted particles if there are enough particles to do so, i.e., if there are at least $H$ particles.
However, if there are $\lceil H / 2 \rceil \leq |\mathcal{P}| < H$ particles (i.e., there are enough particles to close the hull but not enough to fill it with all contracted particles), applying Lemma~\ref{lem:hullfillprogress} iteratively shows that the hull is filled with as many contracted particles as possible.
The following lemma shows that the system terminates correctly in either case.

\begin{lemma} \label{lem:hullfilling}
If there are at least $\lceil H / 2 \rceil$ particles in the system, all particles eventually terminate, filling $H(O)$ with as many contracted particles as possible.
\end{lemma}
\begin{proof}
Since there are at least $\lceil H / 2\rceil$ particles, the hull will be closed by Lemma~\ref{lem:hullclosing}.
If the system contains $|\mathcal{P}| \geq H$ particles, then applying Lemma~\ref{lem:hullfillprogress} iteratively shows that the hull is eventually entirely filled with contracted finished particles.
However, the $|\mathcal{P}| - H$ extra particles must also terminate.
Recall that the leader $\ell$ generates the ``all contracted'' token $all_{con}$ before it finishes at the start of the phase, and that this token is passed backwards along the hull over contracted finished particles only.
Thus, we are guaranteed that every particle from the particle currently holding the $all_{con}$ token up to the finished particle that was the leader is contracted and finished.
Since the hull is eventually filled with all contracted finished particles, $all_{con}$ eventually completes its traversal of the hull, triggering termination that is broadcast to all particles in the system.

If the system contains $\lceil H / 2 \rceil \leq |\mathcal{P}| < H$ particles, then there are too few particles to fill $H(O)$ with only contracted particles.
Thus, applying Lemma~\ref{lem:hullfillprogress} iteratively shows that eventually all particles join the hull and become finished.
Therefore, in all cases, the hull is filled with as many contracted particles as possible and all particles eventually finish.
\end{proof}

We summarize our correctness results in the following theorem.

\begin{theorem} \label{thm:distributedcorrect}
The Convex Hull Algorithm correctly solves instance $(\mathcal{P}, O)$ of the convex hull formation problem if $|\mathcal{P}| \geq |H(O)|$, and otherwise forms a maximal partial strong $\Otri$-hull of $O$.
\end{theorem}

\subsection{Runtime Analysis} \label{subsec:multiruntime}

We now bound the worst-case number of asynchronous rounds for the leader $\ell$ to learn and form the convex hull.
As in Section~\ref{sec:counter}, we use dominance arguments to show that the worst-case number of parallel rounds required by a carefully defined parallel schedule is no less than the runtime of our algorithm.
The first dominance argument will show that the counters bits are forwarded quickly enough to avoid blocking leader expansions and role-intos.
The second will relate the time required for $\ell$ to traverse the object's boundary and convex hull to the running time of our algorithm.
Both build upon \ifconf\else previous work~\cite{Daymude2018}\fi, which analyzed spanning trees of particles led by their root particles.
Several nontrivial extensions are needed here to address the interactions between the counters and particle movements as well as traversal paths that can be temporarily blocked.

We first analyze the performance of bit forwarding (Section~\ref{subsec:counteradapt}).
Note that this is independent of the actual counter operations analyzed in Section~\ref{sec:counter}; here, we analyze how the particles forward their counter bits towards the leader $\ell$.
Suppose a counter $d_h$ is maintained by particles $\ell = P_0, P_1, \ldots P_k$, that is, each particle $P_i$ holds 1--2 bits of $d_h$ and particle $P_k$ holds the final token $f_h$.
A \emph{bit forwarding configuration} $C$ of counter $d_h$ encodes the number of counter elements (i.e., bits of $d_h$ or the final token $f_h$) each particle holds as $C = [C(0), \ldots, C(k)]$, where $C(i) \in \{1, 2\}$ is the number of elements held by particle $P_i$.
A bit forwarding configuration $C$ \emph{dominates} another configuration $C'$ --- denoted $C \succeq C'$ --- if and only if the first $i$ particles of $C$ hold at least as many bits of $d_h$ as the first $i$ particles of $C'$ do, i.e., if $\sum_{j=0}^iC(j) \geq \sum_{j=0}^iC'(j)$ for all $i \in \{0, \ldots, k\}$.

\begin{definition} \label{def:bitschedule}
A \emph{parallel bit forwarding schedule} $(C_0, \ldots, C_T)$ is a sequence of bit forwarding configurations such that for every $t \in \{0, \ldots, T\}$, $C_{t+1}$ is reached from $C_t$ such that one of the following holds for each particle $P_i$, where $i \in \{0, \ldots, k\}$:
\begin{enumerate}
\item Particle $P_i$ does not forward or receive any bits, so $C_{t+1}(i) = C_t(i)$.
\item The leader $\ell = P_0$ performs a role-swap with a particle in front of it, say $P_{-1}$, so $C_{t+1}(0) = C_t(0) - 1 = 1$ and $C_{t+1}(-1) = 1$, shifting the indexes forward.
\item Particle $P_k$ holding $f_h$ either forwards $f_h$ to $P_{k-1}$ or $P_{k-1}$ takes $f_h$ from $P_k$, so $C_{t+1}(k) = C_t(k) - 1 = 0$ and $C_{t+1}(k-1) = C_t(k-1) + 1 = 2$.
\item Particle $P_i$ forwards a counter element to $P_{i-1}$ and takes a counter element from $P_{i+1}$, so $C_{t+1}(i+1) = C_t(i+1) - 1$, $C_{t+1}(i) = C_t(i) = 1$, and $C_{t+1}(i-1) = C_t(i-1) + 1 = 2$.
\end{enumerate}
Such a schedule is \emph{greedy} if the above actions are taken whenever possible without disconnecting the counters (i.e., leaving some $C(i) = 0$ for $i < k$) or giving any particle more than two elements.
\end{definition}

Greedy parallel bit forwarding schedules can be directly mapped onto the \emph{greedy parallel (movement) schedules} of~\cite{Daymude2018}.
A particle keeping two counter elements in a bit forwarding configuration $C$ corresponds to a contracted particle in a particle system configuration $M$; two adjacent particles each keeping a single counter element in $C$ correspond to a single expanded particle in $M$.
When mapped this way, Definition~\ref{def:bitschedule} corresponds exactly to the definition of a parallel movement schedule in~\cite{Daymude2018}.
In fact, the way we forward and take tokens (Section~\ref{subsec:learningalg}) can be exactly mapped onto expansions, contractions, and handovers of particles.
So the next result follows immediately from Lemmas 2 and 3 of~\cite{Daymude2018} and the fact that $\ell$ can only role-swap if it has two counter elements.

\begin{lemma} \label{lem:swapunblock}
Suppose leader $\ell$ only has one bit of a counter $d_h$ and is not holding the final token $f_h$ in round $0 \leq t \leq T - 2$ of greedy parallel bit forwarding schedule $(C_0, \ldots, C_T)$.
Then within the next two parallel rounds, $\ell$ will either have a second bit of $d_h$ or will be holding $f_h$.
\end{lemma}

Next, we combine the parallel counter schedule of Definition~\ref{def:counterschedule}, the parallel bit forwarding schedule of Definition~\ref{def:bitschedule}, and the movements of particles following leader $\ell$ to define a more general \emph{parallel tree-path schedule}.
We use these parallel tree-path schedules to bound the runtime of a spanning tree of particles led by a leader traversing some path $L$.
This bound will be the cornerstone of our runtime proofs for the learning and formation phases.
Here, we consider particle system configurations $C$ that encode each particle's position, state, whether it is expanded or contracted, and any counter bits and tokens it may be holding.
Note that $C$ contains all the information encoded by the counter configurations of Definition~\ref{def:counterschedule} and by the bit forwarding configurations of Definition~\ref{def:bitschedule}.
Thus, for a particle system configuration $C$, let $C^{count}$ (resp., $C^{bit}$) be the counter configuration (resp., bit forwarding configuration) based on $C$.

\begin{definition} \label{def:treepathschedule}
A \emph{parallel tree-path schedule} $((C_0, \ldots, C_T), L)$ is a sequence of particle system configurations $(C_0, \ldots, C_T)$ such that the particle system in $C_0$ forms a tree of contracted particles rooted at the leader $\ell$, $L$ is a (not necessarily simple) path in $\Gtri \setminus O$ starting at the position of $\ell$ in $C_0$ and, for every $t \in \{0, \ldots, T\}$, $C_{t+1}$ is reached from $C_t$ such that $(i)$ any counter operations are processed according to the parallel counter schedule $(\Delta, (C_0^{count}, \ldots, C_i^{count}))$ where $\Delta$ is the sequence of counter operations induced by the change vectors $\delta_i$ associated with $L$, $(ii)$ any bit forwarding operations are processed according to the parallel bit forwarding schedule $(C_0^{bit}, \ldots, C_i^{bit})$, and $(iii)$ one of the following hold for each particle $P$:
\begin{enumerate}
\item The next position in path $L$ is occupied by a particle and particle $P = \ell$ role-swaps with it.
\item The next position in path $L$ is unoccupied and particle $P = \ell$ expands into it.
\item Particle $P$ contracts, leaving the node occupied by its tail empty in $C_{t+1}$.
\item Particle $P$ performs a handover with a neighbor $Q$.
\item Particle $P$ does not move, occupying the same nodes in $C_t$ and $C_{t+1}$.
\end{enumerate}
Such a schedule is \emph{greedy} if the parallel counter and bit forwarding schedules are greedy and the above actions are taken whenever possible without disconnecting the particle system or the counters.
\end{definition}

Even when $L$ is not a simple path, we know the distributed binary counters never disconnect or intersect by Lemma~\ref{lem:counterconnect} and Corollary~\ref{cor:counterintersect}.
Thus, for any greedy parallel counter schedule, its greedy parallel counter and bit forwarding schedules are characterized by Theorem~\ref{thm:counterruntime} and Lemma~\ref{lem:swapunblock}, respectively.
Property 1 of Definition~\ref{def:treepathschedule} handles role-swaps.
Recall that if the leader $\ell$ is contracted, it must either hold two bits or the final token of each of its counters in order to role-swap with a particle blocking its traversal path without disconnecting the counters.
So by Lemma~\ref{lem:swapunblock}, $\ell$ is never waiting to perform a role-swap for longer than a constant number of rounds in the parallel execution.
The remaining properties are exactly those of a parallel (movement) schedule defined in~\cite{Daymude2018}.
Thus, by Lemmas 3 and 9 of~\cite{Daymude2018},
\ifconf
\else
we have the following result:

\begin{lemma} \label{lem:pathruntime}
If $L$ is the (not necessarily simple) path of the leader's traversal, the leader traverses this path in $\mathcal{O}(|L|)$ asynchronous rounds in the worst case.
\end{lemma}
\fi

Using Lemma~\ref{lem:pathruntime}, we can directly relate the distance the leader $\ell$ has traversed to the system's progress towards learning and forming the convex hull.
Once again, recall that $B = |B(O)|$ is the length of the object's boundary and $H = |H(O)|$ is the length of the object's convex hull.
By Lemma 5 of~\cite{Daymude2018}, $B$ particles self-organize as a spanning tree rooted at $\ell$ in at most $\mathcal{O}(B)$ asynchronous rounds.
By Lemmas~\ref{lem:estimate} and~\ref{lem:termination}, $\ell$ traverses $B(O)$ at most twice before completing the learning phase.
Thus, by Lemma~\ref{lem:pathruntime}:

\begin{lemma} \label{lem:learningruntime}
The learning phase completes in at most $\mathcal{O}(B)$ asynchronous rounds.
\end{lemma}

The analysis of the hull closing subphase is similar, but contains an additional technical detail: the condition that a marker particle can only contract as part of a handover is not represented in a greedy parallel tree-path schedule.
In particular, Property 3 of Definition~\ref{def:treepathschedule} says that a marker particle with no children should contract in a greedy parallel tree-path schedule since doing so does not disconnect the particle system or the counters.
But doing so would vacate the leader's starting position, which is explicitly prohibited by the algorithm.
We could define yet another type of parallel schedule capturing this condition and carefully relate it to the parallel movement schedules of~\cite{Daymude2018}, but ultimately all this effort would only show that the leader progresses according to Lemma~\ref{lem:pathruntime} regardless of this discrepancy until all particles are expanded, at which point no particles can move any further.
Thus, because the analysis is technical without providing any new insights, we simply claim that Lemma~\ref{lem:pathruntime} also holds in the presence of an expanded marker particle with no children as long as there exists at least one contracted particle.
With this observation, we can prove the following runtime bound for the hull closing subphase.

\begin{lemma} \label{lem:hullcloseruntime}
In at most $\mathcal{O}(H)$ asynchronous rounds from when the leader $\ell$ completes the learning phase, either $\ell$ completes its traversal of $H(O)$ and closes the hull or every particle in the system terminates, expanded over two nodes of $H(O)$.
\end{lemma}
\begin{proof}
If there are sufficiently many particles to close the hull (i.e., $|\mathcal{P}| \geq \lceil H/2 \rceil$), then $\ell$ will complete its traversal of $H(O)$ by Lemma~\ref{lem:hullclosing}.
The length of this traversal is $H$, so by Lemma~\ref{lem:pathruntime}, $\ell$ closes the hull in at most $\mathcal{O}(H)$ asynchronous rounds.

If instead there are insufficient particles to close the hull (i.e., $|\mathcal{P}| < \lceil H/2 \rceil$), then Lemma~\ref{lem:insufficienthullclose} shows that every particle will eventually be expanded, occupying nodes of $H(O)$.
Until all particles become expanded, there must exist at least one contracted particle in the system.
The length of the leader's traversal path in this case is $2 \cdot |\mathcal{P}| < H$, so by Lemma~\ref{lem:pathruntime}, all particles become expanded in at most $\mathcal{O}(H)$ asynchronous rounds.
Whenever it was that the marker particle first became expanded and had no children, it generated the ``all expanded'' token $all_{exp}$.
Once all particles are expanded, the $all_{exp}$ token must be passed forward at least once per asynchronous round, so the $all_{exp}$ token reaches $\ell$ in at most another $\mathcal{O}(|\mathcal{P}|) = \mathcal{O}(H)$ asynchronous rounds.
In a similar fashion, it takes at most another $\mathcal{O}(|\mathcal{P}|) = \mathcal{O}(H)$ asynchronous rounds for $\ell$ to broadcast termination to all particles in the system.
Thus, every particle in the system is terminated and expanded over two nodes of $H(O)$ in at most $\mathcal{O}(H)$ asynchronous rounds.
\end{proof}

The final runtime lemma analyzes the hull filling subphase.

\begin{lemma} \label{lem:hullfillruntime}
In at most $\mathcal{O}(H)$ asynchronous rounds from when the leader $\ell$ closes the hull, $\min\{|\mathcal{P}|, H\}$ nodes of $H(O)$ will be filled with contracted finished particles.
\end{lemma}
\begin{proof}
The main idea of this argument is to define a potential function representing the system's progress towards filling the hull with contracted finished particles, and then argue that this potential function reaches its maximum --- representing the system filling the hull with as many contracted finished particles as possible --- in at most $\mathcal{O}(H)$ asynchronous rounds.
We consider \emph{parallel filling schedules} that work similarly to parallel path-tree schedules, but take into consideration all the state transitions and movement rules of the hull filling subphase.
Define a \emph{filler segment} $F$ as a connected sequence of filler particles on the surface of the hull, and the \emph{head} $h(F)$ of a filler segment as be filler furthest clockwise in the segment.
At time step $t$, let $U_t \subseteq H(O)$ be the set of nodes not yet occupied by contracted finished particles, $f_t$ be the number of (pre-)filler particles in the system, and $\mathcal{F}_t$ be the set of distinct filler segments on the surface of the hull.
We define our potential function as $\Phi(t) = -|U_t| + f_t - d(\mathcal{F}_t)$, where $d$ is a function that sums the length of each traversal path from the head of a filler segment $F \in \mathcal{F}_t$ to the node in $U(t)$ it eventually fills.
We then argue that $\Phi(t)$ strictly increases every constant number of rounds until either $U_t = 0$, meaning all hull nodes are occupied by contracted finished particles, or $U_t = 2(H - |\mathcal{P}|)$, meaning there were insufficient particles to fill every hull node with a contracted finished particle.
\end{proof}

Putting it all together, we know the algorithm is correct by Theorem~\ref{thm:distributedcorrect}, the learning phase terminates in $\mathcal{O}(B)$ asynchronous rounds by Lemma~\ref{lem:learningruntime}, the hull closing subphase terminates in an additional $\mathcal{O}(H)$ asynchronous rounds by Lemma~\ref{lem:hullcloseruntime}, and the hull filling subphase fills the convex hull with as many contracted particles as possible in another $\mathcal{O}(H)$ asynchronous rounds by Lemma~\ref{lem:hullfillruntime}.
Thus, since $B \geq H$, we complete our analysis with the following theorem.

\begin{theorem} \label{thm:runtime}
In at most $\mathcal{O}(B)$ asynchronous rounds, the Convex Hull Algorithm either solves instance $(\mathcal{P}, O)$ of the convex hull formation problem if $|\mathcal{P}| \geq |H(O)|$ or forms a maximal partial strong $\Otri$-hull of $O$ otherwise.
\end{theorem}

The time required for all particles in the system to terminate may be longer than the bound given in Theorem~\ref{thm:runtime}, depending on the number of particles.
As termination is further broadcast to the rest of the system, we know that at least one non-finished particle receives a termination signal and becomes finished in each asynchronous round.
So,

\begin{corollary} \label{cor:finaltermination}
All particles in system $\mathcal{P}$ terminate the Convex Hull Algorithm in $\mathcal{O}(|\mathcal{P}|)$ asynchronous rounds in the worst case.
\end{corollary}

\section{Forming the (Weak) \texorpdfstring{$\Otri$}{O}-Hull} \label{sec:weakhull}

To conclude, we show how the Convex Hull Algorithm can be extended to form the (weak) $\Otri$-hull of object $O$, solving the $\Otri$-hull formation problem.
Our algorithm, which we refer to as the \emph{$\Otri$-Hull Algorithm}, extends the Convex Hull Algorithm at the point when a finished particle (say, $P_\first$) first holds the ``all contracted'' token with counter value $7$ and would usually broadcast termination.
Note that by Theorem~\ref{thm:distributedcorrect} this only happens if $|\mathcal{P}| \geq H$, which for the $\Otri$-Hull Algorithm we assume to be true.
Instead of terminating, $P_\first$ initiates the $\Otri$-hull formation phase of our algorithm by becoming \emph{tightening}.
Every finished contracted particle whose parent $P.\parent$ is tightening becomes tightening as well, declaring $P.\parent$, which must be the next particle clockwise from $P$ on $H(O)$, as its \emph{successor}; analogously, the \emph{predecessor} of $P$ will be the particle $Q$ on $H(O)$ such that $Q.\parent = P$.
Any particle $P$ that is not finished becomes \emph{non-tightening}, if it has a tightening or non-tightening particle $Q$ in its neighborhood, and sets $P.\parent = Q$.
As the outcome, the particles on $H(O)$ form a bi-directed cycle of contracted tightening particles, and all other particles are non-tightening, their parent pointers forming a spanning forest in which each root is a tightening particle.

Throughout the algorithm, we say a tightening particle $P$ is \emph{convex} (resp., \emph{reflex}) if $P$ and its successor and predecessor are tightening and contracted, its successor lies in direction $d$, and its predecessor lies in direction $(d + 2) \text{ mod } 6$ (resp., $(d + 4) \text{ mod } 6$); see \figtext~\ref{fig:ohullprogress:a}.
The idea of our algorithm is to progressively transform the structure of contracted finished particles initially forming the convex hull $H(O)$ into the object's $\Otri$-hull $H'(O)$ by repeatedly moving convex particles towards the object (see \figtext~\ref{fig:ohullprogress:b}--\ref{fig:ohullprogress:d}).

\begin{figure}[t]
\centering
\begin{subfigure}{.265\textwidth}
	\centering
	\includegraphics[scale=.65]{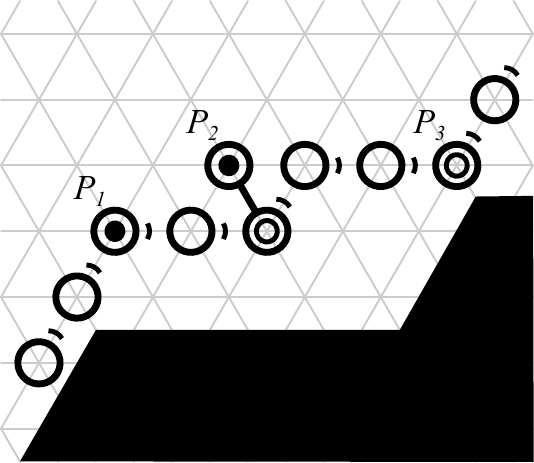}
	\caption{\centering}
	\label{fig:ohullprogress:a}
\end{subfigure}
\begin{subfigure}{.23\textwidth}
	\centering
	\includegraphics[scale=.65]{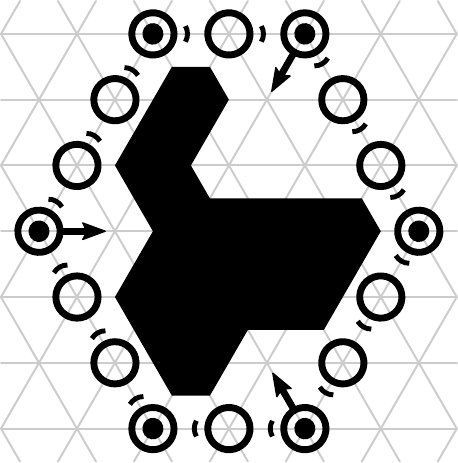}
	\caption{\centering}
	\label{fig:ohullprogress:b}
\end{subfigure}
\begin{subfigure}{.23\textwidth}
	\centering
	\includegraphics[scale=.65]{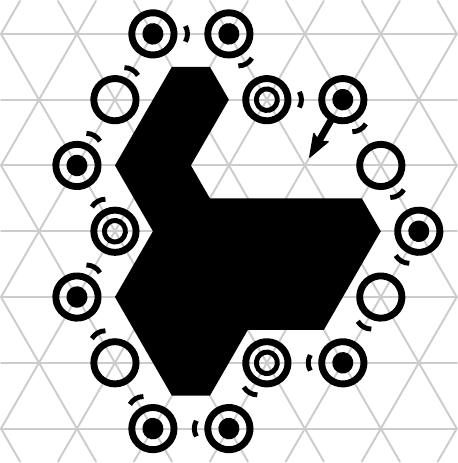}
	\caption{\centering}
	\label{fig:ohullprogress:c}
\end{subfigure}
\begin{subfigure}{.23\textwidth}
	\centering
	\includegraphics[scale=.65]{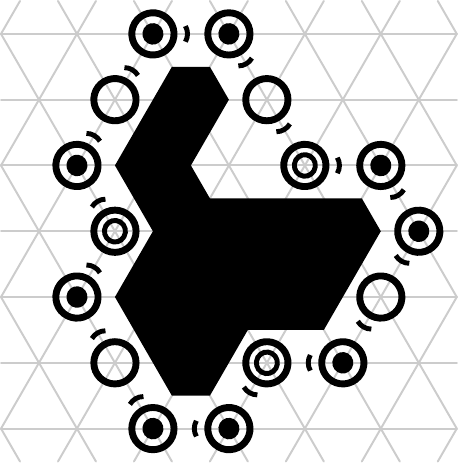}
	\caption{\centering}
	\label{fig:ohullprogress:d}
\end{subfigure}
\caption{Some movements of tightening particles.
(a) $P_1$ is convex (black dot) and can perform a movement into the node between its successor and predecessor, $P_2$ has just made such a move, and $P_3$ is reflex (black circle).
(b)--(d) Transforming the cycle of tightening particles that initially form $H(O)$ into $H'(O)$ by repeatedly moving convex particles towards the object.}
\label{fig:ortho-convex}
\end{figure}

\paragraph*{Moving Convex Particles}

Since the cycle of finished particles initially occupies the convex hull $H(O)$, the algorithm begins with exactly six convex particles and no reflex particles.
Whenever a contracted convex particle $P$ becomes activated, it moves into the node ``between'' its successor and predecessor, i.e., into the node $v$ in direction $(d + 1) \text{ mod } 6$, where $d$ is the direction to its successor (see particle $P_2$ in \figtext~\ref{fig:ohullprogress:a}).
Note that $v$ must be a node contained within the strong $\mathcal{O}$-hull of $O$, i.e., moving $P$ ``shrinks'' the cycle of tightening particles towards $H'(O)$.
More specifically, if $v$ is unoccupied, $P$ simply expands into $v$.
Otherwise, $v$ must be occupied by a particle $Q$ that is \emph{non-tightening}.
If $Q$ is expanded, $P$ simply pushes $Q$.
Otherwise, it role-swaps with $Q$ by declaring $Q$ to be a tightening particle, demoting itself to a non-tightening particle, setting $Q$ as its parent, and updating the predecessor and successor relationships for $Q$ while erasing its own.

If $P$ is expanded, it pulls a contracted non-tightening child in a handover, if one exists, and otherwise contracts.
Note that, as in the hull filling subphase of the Convex Hull Algorithm, the distributed binary counters are no longer in use.
Thus, any potential handovers can be performed without regard for the connectivity of the counters.

\paragraph*{Termination Detection}

Finally, we describe how to detect when the $\Otri$-hull has been formed.
When for the first time $P_\first$ occupies a node adjacent to $O$, it sends a \emph{tight-termination} token with value $1$ and forwards it to its successor.
If a convex particle $P$ has this token and can perform a movement, it sets the token value to $0$ before forwarding it to its successor; by all other particles, the token is forwarded without any value change.
If $P_\first$ receives the token with value $0$, it deduces that there are still movements being made and resets the token value to $1$, once again forwarding it to its successor.
But if $P_\first$ receives the token with value $1$, it knows the $\Otri$-hull has been constructed.
So $P_\first$ terminates by becoming \emph{tight-finished}, and any contracted particle with a tight-finished neighbor also becomes tight-finished.

\subsection{Analysis} \label{subsec:weakhullanalysis}

We now show the correctness and runtime of the $\Otri$-Hull Algorithm.
Recall that the $\Otri$-hull $H'(O)$ has been formed if all nodes of $H'(O)$ are occupied by contracted particles.

\begin{lemma} \label{lem:ortho-premature}
$P_\first$ does not become tight-finished before $H'(O)$ has been formed.
\end{lemma}
\begin{proof}
Let $P_0 = P_\first$ and $C = (P_0, P_1, \ldots, P_m = P_0)$ be the cycle of tightening particles, where $P_{i+1}$ is the successor of $P_i$.
Note that $C$ never changes during the execution of the algorithm (by relabeling the particles involved in a role-swap).
Also observe that if a contracted particle $P_i$ cannot perform a movement at time $t$ but can perform a movement at time $t' > t$, then $P_{i-1}$ or $P_{i+1}$ must perform a movement at some time between $t$ and $t'$.

Suppose to the contrary that $P_0 = P_\first$ becomes tight-finished at time $t^*$ although a movement of some convex particle is still possible.
Then the tight-termination token must have traversed the whole cycle, returning to $P_\first$ with value $1$ at time $t^*$.
Let time $t \leq t^*$ be the earliest time at which some particle $P_i$ with $0 < i < k$ holds the tight-termination token and a particle $P_j$ with $j < i$ can perform a movement; informally, $t$ is the first time that a movement appears ``behind'' the tight-termination token's sweep of cycle $C$ as it searches for movements.
Particle $P_0 = P_\first$ can never perform a movement, as it is already adjacent to the object at the time it creates the tight-termination token, and no tightening particle adjacent to the object can ever perform a movement.
By the minimality of $t$, we know that for any $0 \leq t' < t$, all particles from $P_1$ up to the particle holding the tight-termination token also cannot perform a movement; in particular, this is true at time $t-1$.
Thus, every particle $P_k$ with $0 \leq k \leq i$, including $P_{j-1}$ and $P_{j+1}$, cannot perform a movement at time $t-1$.
But by the observations made above, this yields a contradiction to the claim that $P_j$ could perform a movement at time $t$.
\end{proof}

\begin{lemma} \label{lem:ortho-correct}
The tightening particles eventually form $H'(O)$, after which no convex particle can move anymore.
\end{lemma}
\begin{proof}
Let $U_i \subset V$ be the set of nodes enclosed by the cycle of tightening particles after the $i$-th movement (not containing nodes occupied by tightening particles).
To show that the tightening particles eventually form $H'(O)$, we show the following two claims:
(1) $U_i$ is $\mathcal{O}$-convex and contains the object $O$ for all $i$.
(2) If $U_i$ is an $\mathcal{O}$-convex set containing $O$, but is not minimal, then a movement is possible.
Together with the fact that $U_{i+1} \subset U_i$, this proves that $H'(O)$ is eventually formed.
Clearly, once $U_i$ is minimal, no movement is possible anymore, as otherwise it could not have been minimal.

To prove (1), argue by induction on $i$.
Initially, the particles form $H(O)$, so $U_0$ is the strong $\mathcal{O}$-hull of $O$, and, by definition, is $\mathcal{O}$-convex and contains $O$.
Now suppose $U_i$ is $\mathcal{O}$-convex and contains $O$ by the induction hypothesis, and let $P$ be the particle that performs the next movement into a node $v$ in direction $(d+1) \text{ mod } 6$.
Clearly, $U_{i+1}$ still contains $O$.
Since $U_i$ is $\mathcal{O}$-convex, the intersection of $U_i$ with any straight line of nodes containing $v$ is connected.
These intersections remain connected in $U_{i+1}$ since the neighbors of $v$ in directions $(d+3) \text{ mod } 6$, $(d+4) \text{ mod } 6$, and $(d+5) \text{ mod } 6$ are not in $U_i$ and $U_{i+1} = U_i \setminus \{v\}$.

To prove (2), assume that no movement of a convex particle is possible anymore.
Therefore, every convex particle must be adjacent to the object, as only then it is incapable of moving any further.
Every node $v \in U_i \setminus O$ on the boundary of $U_i$ (i.e., that is adjacent to a node of $V \setminus U_i$) therefore lies on a straight line connecting two nodes of $O$ (which are adjacent to convex particles).
Therefore, the set that results from removing $v$ from $U_i$ cannot be $\mathcal{O}$-convex.
Thus, $U_i$ is minimal.
\end{proof}

Taken together, these lemmas prove the correctness of the $\Otri$-Hull Algorithm.
By Lemma~\ref{lem:ortho-premature}, we have that $P_\first$ will not terminate prematurely, stopping the remaining particles from correctly forming $H'(O)$.
So by Lemma~\ref{lem:ortho-correct}, $H'(O)$ is eventually formed and there are no remaining movements.
Thus, the tight-termination token will never be set to value $0$ again, resulting in the following concluding lemma.

\begin{lemma} \label{lem:ortho-token}
Once $H'(O)$ has been formed, the tight-termination token traverses the cycle at most twice before $P_\first$ terminates.
\end{lemma}

We now turn to the runtime analysis.
Recall that $H = |H(O)| = |H'(O)|$ and $n = |\mathcal{P}|$.
As in Section~\ref{subsec:multiruntime}, we first bound the runtime for a synchronous parallel execution of the algorithm, and then argue that the execution is dominated by our asynchronous algorithm.
As before, we consider particle system configurations $C$ that encode each particle's position, state, whether it is expanded or contracted, and any tokens held by it.

\begin{definition} \label{def:cycleschedule}
A \emph{parallel cycle schedule} is a schedule $(C_0, \ldots, C_t)$ such that in $C_0$ all nodes of $H(O)$ are occupied by contracted tightening particles forming a directed cycle in clockwise direction, and all other particles are contracted and connected to a tightening particle via a sequence of parent pointers.
For every $0 \leq i < t$, $C_{i+1}$ is reached from $C_i$ such that the following holds for every particle $P$:
\begin{enumerate}
\item $P$ is tightening, has a successor in direction $d$, and moves into the node $u$ in direction $(d + 1) \text{ mod } 6$ by performing an expansion, if $u$ is unoccupied, or by performing a role-swap, otherwise.
\item $P$ has no children and contracts, leaving the node occupied by its tail empty in $C_{i+1}$.
\item $P$ pulls in a child, or, if $P$ is non-tightening, pushes its parent.
\item $P$ occupies the same nodes in $C_i$ and $C_{i+1}$.
\end{enumerate}
Such a schedule is \emph{greedy} if the above actions are taken whenever possible.
\end{definition}

Given any greedy parallel cycle schedule, we can now show the following lemma.

\begin{figure}[t]
\centering
\begin{subfigure}{.24\textwidth}
	\centering
	\includegraphics[scale=.6]{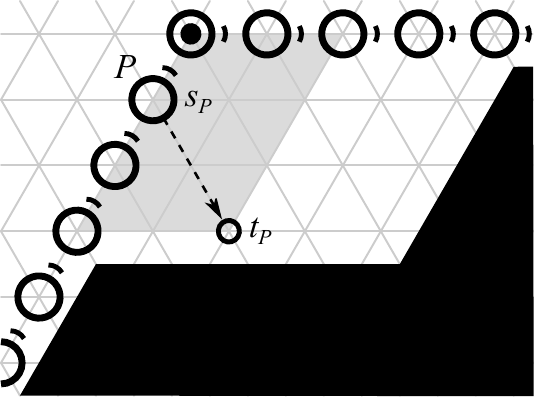}
	\caption{\centering}
	\label{fig:ohullruntime:a}
\end{subfigure}
\hfill
\begin{subfigure}{.24\textwidth}
	\centering
	\includegraphics[scale=.6]{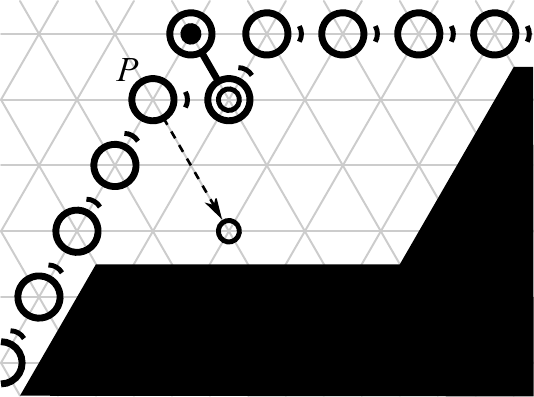}
	\caption{\centering}
	\label{fig:ohullruntime:b}
\end{subfigure}
\hfill
\begin{subfigure}{.24\textwidth}
	\centering
	\includegraphics[scale=.6]{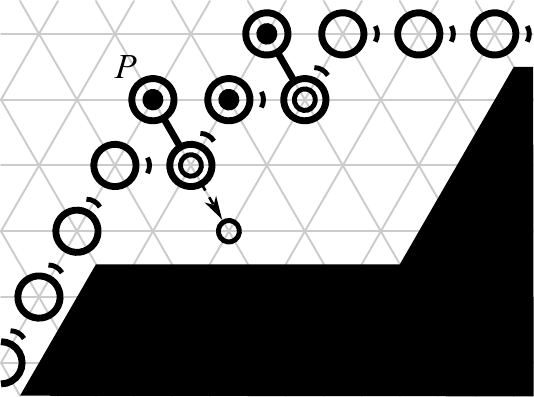}
	\caption{\centering}
	\label{fig:ohullruntime:c}
\end{subfigure}
\hfill
\begin{subfigure}{.24\textwidth}
	\centering
	\includegraphics[scale=.6]{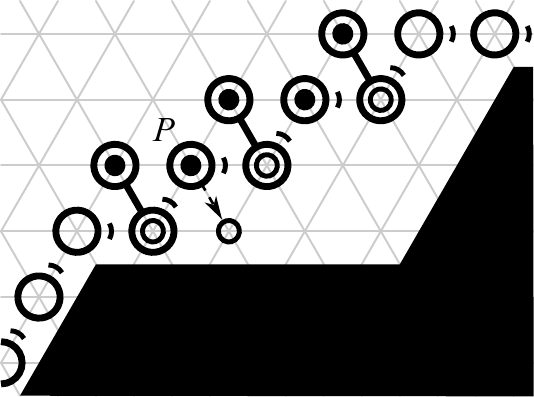}
	\caption{\centering}
	\label{fig:ohullruntime:d}
\end{subfigure}
\caption{A parallel execution of the $\Otri$-Hull Algorithm as described in Lemma~\ref{lem:ortho-runtime-synchro}. (a) A particle $P$ and its starting and ending positions. (b)--(d) Convex particles greedily moving towards the object whenever possible.}
\label{fig:ohullruntime}
\end{figure}

\begin{lemma} \label{lem:ortho-runtime-synchro}
	Any parallel cycle schedule reaches a configuration in which the particles form $H'(O)$ within $\mathcal{O}(H)$ parallel rounds.
\end{lemma}

\begin{proof}
	Consider the structure of tightening particles forming $H(O)$ at the beginning of the algorithm.
	Note that for any two convex particles $P_1$ and $P_2$ that are connected by a straight line of hull particles (i.e., that are visited consecutively in a traversal of the cycle), there exists at least one particle $Q$ between $P_1$ and $P_2$ on that line that is adjacent to the object $O$.
	Any movement that is performed by any particle between $P_1$ and $Q$ on that line can only be a direct or indirect consequence of $P_1$'s first movement, but must be fully independent of any movement of $P_2$.
	Therefore, it suffices to analyze the execution of the algorithm on each of the six "corners" of $H(O)$, i.e., on all tightening particles that initially lie between a convex particle $P$ and the first particle adjacent to the object in any of the two directions in which $P$ has adjacent tightening particles.

    Consider a convex particle $P$ occupying node $u$ at the beginning of the algorithm.
    Let $v$ and $w$ be the first nodes in directions $d$ and $(d + 2) \text{ mod } 6$, respectively, in which $P$ has adjacent tightening particles such that $v$ and $w$ are adjacent to $O$.
	For any hull particle $P$ occupying a node between $u$ and $v$ or between $u$ and $w$, let $i_P$ be the distance from $P$'s initial position $s_P$ to $u$, and let $d_P$ be the distance from $s_P$ to its final position $t_P$ adjacent to $O$ in direction $(d + 1) \text{ mod } 6$.
	We will show that $P$ reaches $t_P$ after at most $i_P + 2 d_P$ synchronous rounds, which, as both $d_i$ and $i$ are bounded above by $H$ and the tightening particles of the six corners of $H(O)$ form $H'(O)$ independently, immediately implies the claim.

	First, note that our algorithm ensures that there is never an expanded non-tightening particle whose parent is expanded, i.e., handovers can always be performed.
	Now fix some particle $P$ occupying a node between $u$ and $v$ or between $u$ and $w$ and consider the parallelogram formed by diagonal vertices $u$ and $t_P$, with edges extending in directions $d$ and $(d + 2) \text{ mod } 6$ (see \figtext~\ref{fig:ohullruntime:a}).
	By the definition of $\Otri$-hull, no node of this parallelogram can be a node of $O$.
	Therefore, the following can easily be shown by induction on the number of synchronous rounds $t \le i_P$:
	Every particle $Q$ such that $i_{Q} \leq t - 1$ is contracted and occupies the node that lies $(t - i_Q) / 2$ steps in direction $(d + 1) \text{ mod } 6$ of $t_Q$, if $i_Q$ and $t$ are both even or both odd;
	otherwise, $Q$ is expanded, its head occupying the node that lies $(t - i_Q + 1) / 2$ steps in direction $(d + 1) \text{ mod } 6$ of $t_Q$ (see \figtext~\ref{fig:ohullruntime:b}--\ref{fig:ohullruntime:d}).
	Therefore, after $i_P$ rounds, the successor of $P$, if $s_P$ lies between $u$ and $v$, or its predecessor, if it lies between $u$ and $w$, is expanded for the first time (as in \figtext~\ref{fig:ohullruntime:b}).
	In the next round, $P$ will perform its first movement.
	It can easily be seen that $P$ will not be hindered in its movement until it reaches $t_P$, which therefore takes $2d_P$ additional rounds.
\end{proof}

Similarly to the proofs of Sections~\ref{subsec:counterruntime} and \ref{subsec:multiruntime} and the discussion of \cite{Daymude2018}, we compare a greedy parallel cycle schedule with an \emph{asynchronous cycle schedule} $(S,(C_0^A), \ldots, C_t^A))$ given a fair asynchronous activation sequence.
For any two configuration $C$ and $C'$ and a tightening particle $P$, we say that \emph{$C$ dominates $C'$ w.r.t. $P$}, if and only if $P$ has performed at least as many movements in $C$ as in $C'$, and say that \emph{$C$ dominates $C'$} if and only if $C$ dominates $C'$ w.r.t. every particle.
Note that if a movement is possible, it can never be hindered, therefore the discussion of \cite{Daymude2018} implies the following lemma.

\begin{lemma} \label{lem:orthodominance}
	Given any fair asynchronous activation sequence $A$ and some initial configuration $C_0^A$ for the $\Otri$-Hull Algorithm, there exists a greedy parallel cycle schedule $(C_0, \ldots, C_t)$ with $C_0 = C_0^A$ such that $C_i^A \succeq C_i$ for all $0 \leq i \leq t$.
\end{lemma}

After $H'(O)$ is formed, the tight-termination token is passed over the entire cycle at most twice by Lemma~\ref{lem:ortho-token}, which takes at most $\mathcal{O}(H)$ asynchronous rounds.
Finally, once $P_\first$ terminates by becoming tight-finished, in the worst case only one additional particle becomes tight-finished in each subsequent round.
Thus, it may take an additional $\mathcal{O}(n)$ asynchronous rounds in the worst case before all particles terminate.
We conclude the following theorem.

\begin{theorem} \label{thm:ortho-convex-hull}
    In at most $\mathcal{O}(|H(O)|)$ asynchronous rounds, the $\Otri$-Hull Algorithm solves instance $(\system, O)$ of the $\Otri$-hull formation problem if $|\system| \geq |H(O)|$.
    After an additional $\mathcal{O}(|\system|)$ rounds, all particles have terminated.
\end{theorem}

\bibliographystyle{plainurl}
\bibliography{refs}

\newpage

\appendix

\section{Algorithm Pseudocode} \label{app:pseudocode}

In this appendix, we provide detailed pseudocode for the Convex Hull Algorithm.
This includes the adaptation of the distributed binary counters to the application of convex hull formation.

\subsection{A Binary Counter of Particles} \label{subapp:countercode}

In the presentation of the distributed binary counter of particles (Section~\ref{sec:counter}), the particles are arranged in a simple path $P_0, P_1, \ldots, P_k$.
Each particle $P_i$ keeps one bit of the counter, denoted $P_i.\bit$, and can hold up to two tokens (except the leader $P_0$ that holds at most one token) in a queue, denoted $P_i.\tokens = [next, queued]$.
We begin by fully specifying this protocol as it is described in Section~\ref{sec:counter}, without the modifications used for its application to convex hull formation.
When activated, each particle $P_i$ processes its counter tokens as described in Algorithm~\ref{alg:counter:tokens}.

\begin{algorithm}
\caption{Binary Counter: Processing Counter Tokens} \label{alg:counter:tokens}
\begin{algorithmic}[1]
\Function{ProcessCounter}{$P_i$, $P_{i+1}$}
\If {the next non-final token in $P_i.\tokens$ is a $c^+$}
    \If {$P_i.\bit = 0$}
        \State Dequeue $c^+$ from $P_i.\tokens$, delete $c^+$, and set $P_i.\bit \gets 1$.
    \ElsIf {($P_i.\bit = 1$) $\wedge$ ($P_{i+1}.\tokens$ is not full, i.e., has less than two tokens)}
        \State Dequeue $c^+$ from $P_i.\tokens$, enqueue $c^+$ into $P_{i+1}.\tokens$, and set $P_i.\bit \gets 0$.
    \ElsIf {$P_i.\bit = \emptyset$}
        \State Dequeue $f$ from $P_i.\tokens$ and enqueue $f$ into $P_{i+1}.\tokens$.\label{alg:counter:tokens:extend}
        \State Dequeue $c^+$ from $P_i.\tokens$, delete $c^+$, and set $P_i.\bit \gets 1$.
    \EndIf
\ElsIf {the next non-final token in $P_i.\tokens$ is a $c^-$}
    \If {$P_i.\bit = 1$}
        \If {$\neg(P_{i+1}.\bit = 1 \wedge P_{i+1}.\tokens = [c^-]$)}
            \State Dequeue $c^-$ from $P_i.\tokens$, delete $c^-$, and set $P_i.\bit \gets 0$.
            \If {($P_{i+1}.\tokens = [f]$) $\wedge$ ($P_i \neq P_0$)}
                \State Dequeue $f$ from $P_{i+1}.\tokens$, enqueue $f$ into $P_i.\tokens$, and set $P_i.\bit \gets \emptyset$.
            \EndIf
        \EndIf
    \ElsIf {($P_i.\bit = 0$) $\wedge$ ($P_{i+1}.\tokens$ is not full)}
        \State Dequeue $c^-$ from $P_i.\tokens$, enqueue $c^-$ into $P_{i+1}.\tokens$, and set $P_i.\bit \gets 1$.
    \EndIf
\EndIf
\EndFunction
\end{algorithmic}
\end{algorithm}

In addition to processing its counter tokens, the leader particle $P_0$ can generate new increment and decrement tokens as well as zero-test the counter (Algorithm~\ref{alg:counter:leadops}).
\textsc{Generate}($op$) generates the token corresponding to the specified operation, which is either an increment or a decrement.
\textsc{ZeroTest}($P_0$, $P_1$) checks whether the counter's value is $0$, and is either ``unavailable'' or returns \textsc{True} or \textsc{False}.

\begin{algorithm}
\caption{Binary Counter: Leader Operations} \label{alg:counter:leadops}
\begin{algorithmic}[1]
\Function{Generate}{$op$}
\If {$P_0.\tokens = []$, i.e., $P_0.\tokens$ is empty}
    \If {$op$ is ``increment''} generate $c^+$ and enqueue it into $P_0.\tokens$.
    \ElsIf {$op$ is ``decrement''} generate $c^-$ and enqueue it into $P_0.\tokens$.
    \EndIf
\EndIf
\EndFunction
\end{algorithmic}
\begin{algorithmic}[1]
\Function{ZeroTest}{$P_0$, $P_1$}
\If {($P_1.\bit = 1$) $\wedge$ ($P_1.\tokens = [c^-]$)} \Return ``unavailable''.
\Else {} \Return ($P_1.\tokens = [f]$) $\wedge$ (($P_0.\bit = 0 \wedge P_0.\tokens = []$) $\vee$ ($P_0.\bit = 1 \wedge P_0.\tokens = [c^-]$))
\EndIf
\EndFunction
\end{algorithmic}
\end{algorithm}

\subsection{Binary Counter Adaptation for Convex Hull Formation} \label{subapp:hullcountercode}

As described in Section~\ref{subsec:counteradapt}, a total of four modifications to the general binary counter protocol were needed for the application of convex hull formation.
First, each particle participates in six counters --- one for each half-plane $h \in \mathcal{H}$ --- instead of just one.
Second, each particle keeps up to two bits of each counter instead of one.
These first two modifications are supported with more general notation.
Each particle $P$ keeps the following information in its memory:
\begin{itemize}
\item For each half-plane $h \in \mathcal{H}$, two bits of a counter $d_h$: $P.\bit_h^L$ (the less significant bit) and $P.\bit_h^M$ (the more significant bit) with values in $\{\emptyset, \sqcup, 0, 1\}$.
Recall that a bit with value $\emptyset$ is beyond the most-significant bit, and a bit with value $\sqcup$ is ``blank''.

\item For each half-plane $h \in \mathcal{H}$, two queues of tokens: $P.\tokens_h^L$ (for the less significant bit) and $P.\tokens_h^M$ (for the more significant bit), each of which can hold up to two counter tokens (increment $c_h^+$, decrement $c_h^-$, or final $f_h$).
\end{itemize}

The third modification handles extending the counter unambiguously in the tree structure of the particle system.
In the case that an increment carry over causes the counter to grow by one bit, the counter is always extended along hull particles, if they exist, or along the boundary of the object otherwise.
This logic is handled by \textsc{NextCounterParticle}($h$) for counter $d_h$.
Finally, because a particle may only be keeping one bit of a counter instead of two (i.e., it may have $P.\bit_h^M = \sqcup$), bit forwarding is used to move a counter's bits as far towards the leader as possible.
This operation is handled by \textsc{ForwardBits}($Q$, $h$), where the activated particle $P$ takes the next-most significant bit of counter $d_h$ from $Q$.

\begin{algorithm}
\caption{Binary Counter for Convex Hull: Processing Counter Tokens} \label{alg:hullcounter:tokens}
\begin{algorithmic}[1]
\Function{NextCounterParticle}{$P$, $h$}
\If {$P$ has a child in the spanning tree $Q$ with $Q.\bit_h^L \in \{0,1\}$ or with $Q.\tokens_h^L \neq []$} \Return $Q$.
\ElsIf {$P$ has a child $Q$ with $Q.\mystate \in \{$\emph{hull}, \emph{marker}, \emph{pre-marker}$\}$} \Return $Q$.
\Else {} \Return any child of $P$ that is on the object's boundary.
\EndIf
\EndFunction
\end{algorithmic}
\begin{algorithmic}[1]
\Function{ForwardBits}{$P$, $Q$, $h$}
\If {($P.\bit_h^M = \sqcup$) $\wedge$ ($Q.\bit_h^M \neq \sqcup$)}
    \State Set $P.\bit_h^M \gets Q.\bit_h^L$ and $P.\tokens_h^M \gets Q.\tokens_h^L$.
    \State Set $Q.\bit_h^L \gets Q.\bit_h^M$ and $Q.\tokens_h^L \gets Q.\tokens_h^M$.
    \If {$Q.\bit_h^M \in \{0, 1\}$} $Q.\bit_h^M \gets \sqcup$.
    \EndIf
    \State Set $Q.\tokens_h^M \gets []$.
\EndIf
\EndFunction
\end{algorithmic}
\begin{algorithmic}[1]
\Function{ProcessHullCounters}{ }
\ForAll {$h \in \mathcal{H}$}
    \State Let $Q \gets$ \Call{NextCounterParticle}{$P$, $h$}.
    \State \Call{ForwardBits}{$P$, $Q$, $h$}.
    \If {$P.\bit_h^M = \sqcup$}
        \State Run Algorithm~\ref{alg:counter:tokens}'s \Call{ProcessCounter}{$P.(\cdot)_h^L$, $Q.(\cdot)_h^L$}.
    \Else {}
        \State Run Algorithm~\ref{alg:counter:tokens}'s \Call{ProcessCounter}{$P.(\cdot)_h^M$, $Q.(\cdot)_h^L$}.
        \State Run Algorithm~\ref{alg:counter:tokens}'s \Call{ProcessCounter}{$P.(\cdot)_h^L$, $P.(\cdot)_h^M$}.
    \EndIf
\EndFor
\EndFunction
\end{algorithmic}
\end{algorithm}

As in the general binary counter, the leader particle $\ell = P_0$ can also generate new increment and decrement tokens as well as zero-test the counter (Algorithm~\ref{alg:hullcounter:leadops}).
\textsc{HullGenerate}($i$) generates the increment and decrement tokens associated with the update vector $\delta_i$ (see Section~\ref{sec:soloalg}).
\textsc{HullZeroTest}($h$) checks whether $d_h = 0$, and is either ``unavailable'' or returns \textsc{True} or \textsc{False}.

\begin{algorithm}
\caption{Binary Counter for Convex Hull: Leader Operations} \label{alg:hullcounter:leadops}
\begin{algorithmic}[1]
\Function{HullGenerate}{$i$}
\ForAll {$h \in \mathcal{H}$}
    \If {$\delta_{i,h} = 1$} generate $c_h^+$ and enqueue it into $\ell.\tokens_h^L$.
    \ElsIf {$\delta_{i,h} = -1$} generate $c_h^-$ and enqueue it into $\ell.\tokens_h^L$.
    \EndIf
\EndFor
\EndFunction
\end{algorithmic}
\begin{algorithmic}[1]
\Function{HullZeroTest}{$h$}
\If {$\ell.\bit_h^M \neq \sqcup$} \Return Algorithm~\ref{alg:counter:leadops}'s \Call{ZeroTest}{$\ell.(\cdot)_h^L$, $\ell.(\cdot)_h^M$}.
\Else {} \Return Algorithm~\ref{alg:counter:leadops}'s \Call{ZeroTest}{$\ell.(\cdot)_h^L$, $P_1.(\cdot)_h^L$}.
\EndIf
\EndFunction
\end{algorithmic}
\end{algorithm}

\subsection{The Learning Phase} \label{subapp:learncode}

In the learning phase (Section~\ref{subsec:learningalg}), a particle $P$ can be in one of three states: \emph{leader}, \emph{follower}, or \emph{idle}, denoted $P.\mystate$.
In addition to those listed in the previous section, $P$ keeps the following information in its memory:
\begin{itemize}
\item A direction $P.\parent \in [6] = \{0, 1, \ldots, 5\}$ pointing to its parent in the spanning tree.
We denote the case that $P$ has no parent as $P.\parent = \emptyset$.
In a slight abuse of notation, we also refer to the particle in the direction of $P.\parent$ as $P.\parent$.

\item For each half-plane $h \in \mathcal{H}$, a terminating bit $b_h \in \{0, 1\}$.
\end{itemize}

The pseudocode is written from the perspective of a particle $P$, unless otherwise specified.
We begin with some helper functions (Algorithm~\ref{alg:learn:help}) that are used throughout the learning phase.
\textsc{GetRHR}() returns a direction $i \in [6]$ that leads to the next node in a clockwise traversal of the object's boundary.
\textsc{HandoverIsSafe}($P$, $Q$) checks whether a handover between a contracted particle $P$ and an expanded particle $Q$ is guaranteed to maintain the connectivity of the counters.
\textsc{RoleSwap}($P$, $Q$) copies the memory of $P$ into the memory of $Q$ and then sets $P$ to be a follower of $Q$.

\begin{algorithm}
\caption{Learning Phase: Helper Functions} \label{alg:learn:help}
\begin{algorithmic}[1]
\Function{GetRHR}{ }
\State Let $i \in [6]$ be such that an object occupies the node in direction $i$.
\While {an object occupies the node in direction $i$} $i \gets (i + 5) \text{ mod } 6$.
\EndWhile
\State \Return $i$.
\EndFunction
\end{algorithmic}
\begin{algorithmic}[1]
\Function{HandoverIsSafe}{$P$, $Q$}
\ForAll {$h \in \mathcal{H}$}
    \State Let $a_1 \gets$ ($P.\bit_h^M \neq \sqcup$) $\wedge$ ($Q.\bit_h^M \neq \sqcup$).
    \State Let $a_2 \gets$ ($P.\bit_h^M = P.\bit_h^L = Q.\bit_h^M = Q.\bit_h^L = \emptyset$).
    \State Let $a_3 \gets$ ($P.\bit_h^M = P.\bit_h^L = \emptyset$) $\wedge$ (one of $Q.\tokens_h^M$ or $Q.\tokens_h^L$ contains $f_h$).
    \If {$\neg(a_1 \vee a_2 \vee a_3)$} \Return \textsc{False}.
    \EndIf
\EndFor
\State \Return \textsc{True}.
\EndFunction
\end{algorithmic}
\begin{algorithmic}[1]
\Function{RoleSwap}{$P$, $Q$}
\ForAll {$h \in \mathcal{H}$}
    \State Set $\{Q.\bit_h^M, Q.\bit_h^L\} \gets \{\sqcup, P.\bit_h^L\}$, $Q.\tokens_h^L \gets P.\tokens_h^L$, and $Q.b_h \gets P.b_h$.
    \State Set $P.\bit_h^L \gets P.\bit_h^M$ and $P.\tokens_h^L \gets P.\tokens_h^M$.
    \State Set $P.\bit_h^M \gets \sqcup$ and $P.\tokens_h^M \gets []$ and clear $P.b_h$.
\EndFor
\State Set $Q.\mystate \gets$ \emph{leader} and $P.\mystate \gets$ \emph{follower}.
\State Set $Q.\parent \gets \emptyset$ and $P.\parent \gets Q$.
\EndFunction
\end{algorithmic}
\end{algorithm}

During the learning phase, the leader particle $\ell$ performs a clockwise traversal of the object's boundary, using its distributed binary counters to measure its current distances from each of the six half-planes (Algorithm~\ref{alg:learn:lead}).
It completes its estimation and moves on to the formation phase once it has visited all six half-planes without pushing any of them.

\begin{algorithm}
\caption{Learning Phase: the Leader Particle} \label{alg:learn:lead}
\begin{algorithmic}[1]
\State \Call{ProcessHullCounters}{ }.
\If {$\ell$ is expanded}
    \If {$\ell$ has a contracted follower child $P$ on the object's boundary}
        \State Perform a pull handover with $P$.
    \EndIf
\ElsIf {$\ell$ is contracted}
    \If {$b_h = 1$ for all $h \in \mathcal{H}$}
        \State Estimation complete; go to formation phase.
    \ElsIf {\Call{HullZeroTest}{$h$} is available and $\ell.\tokens_h^L = []$ for all $h \in \mathcal{H}$}
        \State Let $i \gets$ \Call{GetRHR}{ }, and let $v$ be the node in direction $i$.
        \If {($v$ is unoccupied) $\vee$ (a contracted particle $P$ occupies $v$ and $\ell.\bit_h^M \neq \sqcup$ for all $h \in \mathcal{H}$)}
            \State Use \Call{HullZeroTest}{$h$} to construct $\mathcal{H}' = \{h \in \mathcal{H} : \delta_{i,h} = -1 \text{ and } d_h = 0\}$.
            \State \Call{HullGenerate}{$i$}.
            \If {$\mathcal{H}' \neq \emptyset$} set $b_h \gets 0$ for all $h \in \mathcal{H}$.
            \Else {} use \Call{HullZeroTest}{$h$} to set $b_h \gets 1$ for all $h \in \mathcal{H}$ such that $d_h = 0$.
            \EndIf
            \If {$v$ is unoccupied} expand in direction $i$ into node $v$.
            \Else {} \Call{RoleSwap}{$\ell$, $P$}
            \EndIf
        \EndIf
    \EndIf
\EndIf
\end{algorithmic}
\end{algorithm}

Followers simply follow their parents in the spanning tree through handovers, and contract whenever they have no children themselves (Algorithm~\ref{alg:learn:follow}).

\begin{algorithm}
\caption{Learning Phase: Follower Particles} \label{alg:learn:follow}
\begin{algorithmic}[1]
\State \Call{ProcessHullCounters}{ }.
\If {$P$ is expanded}
    \If {$P$ has no children nor any idle neighbors}
        \State Contract tail.
    \ElsIf {$P$ has a contracted child $Q$ for which \Call{HandoverIsSafe}{$Q$, $P$}}
        \State Perform a pull handover with $Q$.
    \EndIf {}
\ElsIf {$P.\parent$ is expanded and \Call{HandoverIsSafe}{$P$, $P.\parent$}}
    \State Perform a push handover with $P.\parent$.
\EndIf
\end{algorithmic}
\end{algorithm}

Finally, idle particles simply attempt to become followers (Algorithm~\ref{alg:learn:idle}).

\begin{algorithm}
\caption{Learning Phase: Idle Particles} \label{alg:learn:idle}
\begin{algorithmic}[1]
\If {$P$ has a neighbor $Q$ with $Q.\mystate \neq$ \emph{idle}}
    \State $P.\mystate \gets$ \emph{follower}.
    \State $P.\parent \gets Q$.
\EndIf
\end{algorithmic}
\end{algorithm}

\subsection{The Hull Closing Subphase} \label{subapp:hullclosecode}

The first of the formation phase's subphases is the hull closing subphase (Section~\ref{subsec:formationalg}).
This subphase introduces four new states: \emph{hull}, \emph{pre-marker}, \emph{marker}, and \emph{finished}.
The leader keeps an additional variable $\ell.\plane \in \mathcal{H}$ that indicates which half-plane boundary it is currently following.
All particles have the potential of encountering two new token types in this subphase: the ``all-expanded'' token $all_{exp}$ and the ``termination'' token.

We first introduce some helper functions (Algorithm~\ref{alg:hullclose:help}) that are used by the leader.
\textsc{PlaneToDir}($h$) returns a direction $i \in [6]$ that leads to the next node in a clockwise traversal of the convex hull along the specified half-plane.
\textsc{NextPlane}($h$) returns the next half-plane from $h$ in a clockwise traversal of the convex hull.
\textsc{RoleSwap}($P$, $Q$), like the learning phase helper function by the same name, copies the memory of $P$ into the memory of $Q$.
However, in the hull closing subphase, a role-swap leaves $P$ as a hull particle and does not bother with the terminating bits $b_h$ that are no longer used.

\begin{algorithm}
\caption{Hull Closing Subphase: Helper Functions} \label{alg:hullclose:help}
\begin{algorithmic}[1]
\Function{PlaneToDir}{$h$}
\State Let $i$ be the index of $h$ in $\mathcal{H} = \{N, NE, SE, S, SW, NW\}$.
\State \Return $(i + 1) \text{ mod } 6$.
\EndFunction
\end{algorithmic}
\begin{algorithmic}[1]
\Function{NextPlane}{$h$}
\State \Return $\mathcal{H}[$\Call{PlaneToDir}{$h$}$]$.
\EndFunction
\end{algorithmic}
\begin{algorithmic}[1]
\Function{RoleSwap}{$P$, $Q$}
\ForAll {$h \in \mathcal{H}$}
    \State Set $\{Q.\bit_h^M, Q.\bit_h^L\} \gets \{\sqcup, P.\bit_h^L\}$ and $Q.\tokens_h^L \gets P.\tokens_h^L$.
    \State Set $P.\bit_h^L \gets P.\bit_h^M$ and $P.\tokens_h^L \gets P.\tokens_h^M$.
    \State Set $P.\bit_h^M \gets \sqcup$ and $P.\tokens_h^M \gets []$.
\EndFor
\State Set $Q.\mystate \gets$ \emph{leader} and $P.\mystate \gets$ \emph{hull}.
\State Set $Q.\parent \gets \emptyset$ and $P.\parent \gets Q$.
\EndFunction
\end{algorithmic}
\end{algorithm}

During the hull closing subphase, the leader leads the rest of the particle system along a clockwise traversal of the convex hull using its binary counters to know when to turn (Algorithm~\ref{alg:hullclose:lead}).
The leader's initial position is kept by a \emph{marker} particle; the leader moves on to the next subphase when it encounters the marker particle again.
However, if there are not enough particles to close the hull ($n < \lceil H / 2\rceil$), the leader will eventually receive the $all_{exp}$ token, triggering an early termination.

\begin{algorithm}
\caption{Hull Closing Subphase: the Leader Particle} \label{alg:hullclose:lead}
\begin{algorithmic}[1]
\State \Call{ProcessHullCounters}{ }.
\If {$\ell$ is holding the $all_{exp}$ token and the marker particle is adjacent to $\ell$ but not a child of $\ell$}
    \State Generate a termination token and pass it to the hull or marker child of $\ell$.
    \State Set $\ell.\mystate \gets$ \emph{finished} and \Return
\EndIf
\If {$\ell$ is expanded}
    \If {$\ell$ has a contracted hull child $Q$} perform a pull handover with $Q$.
    \ElsIf {$\ell$ has no hull children but has a contracted follower child $Q$ keeping counter bits}
        \State Set $Q.\mystate \gets$ \emph{pre-marker}.
        \State Perform a pull handover with $Q$.
    \EndIf
\EndIf
\If {\Call{HullZeroTest}{\textsc{NextPlane}($\ell.\plane$)} $\neq$ ``unavailable''}
    \If {\Call{HullZeroTest}{\textsc{NextPlane}($\ell.\plane$)}}
        \State Set $\ell.\plane \gets$ \Call{NextPlane}{$\ell.\plane$}.
    \EndIf
    \State Let $i \gets$ \Call{PlaneToDir}{$\ell.\plane$}, and let $v$ be the node in direction $i$ from the head of $\ell$.
    \If {$v$ is occupied by the marker particle $Q$}
        \State The hull is closed. Set $Q.\mystate \gets$ \emph{finished} and go on to hull filling subphase.
    \ElsIf {($\ell$ is contracted) $\wedge$ ($\ell.\tokens_h^L = []$ for all $h \in \mathcal{H}$) $\wedge$ (($v$ is unoccupied) $\vee$ (a contracted follower $P$ occupies $v$ and $\ell.\bit_h^M \neq \sqcup$ for all $h \in \mathcal{H}$))}
        \State \Call{HullGenerate}{$i$}
        \If {$v$ is unoccupied} expand in direction $i$ into node $v$.
        \Else {} \Call{RoleSwap}{$\ell$, $P$}.
        \EndIf
    \EndIf
\EndIf
\end{algorithmic}
\end{algorithm}

Hull particles are followers that have entered the convex hull structure.
Specifically, every particle from the marker to the leader is a hull particle.
These particles simply follow the leader, maintain the binary counters, and pass tokens they receive (Algorithm~\ref{alg:hullclose:hull}).

\begin{algorithm}
\caption{Hull Closing Subphase: Hull Particles} \label{alg:hullclose:hull}
\begin{algorithmic}[1]
\State \Call{ProcessHullCounters}{ }.
\If {$P$ is holding the termination token}
    \State Set $P.\mystate \gets$ \emph{finished}.
    \State Pass the termination token to the hull or marker child of $P$.
\ElsIf {$P$ is expanded}
    \If {$P$ has a contracted hull child $Q$}
        \State Perform a pull handover with $Q$.
    \ElsIf {$P$ is holding the $all_{exp}$ token and $P.\parent$ is expanded}
        \State Pass the $all_{exp}$ token to $P.\parent$.
    \EndIf
\ElsIf {$P$ is contracted and $P.\parent$ is expanded}
    \State Perform a push handover with $P.\parent$.
\EndIf
\end{algorithmic}
\end{algorithm}

The pre-marker particle is always expanded and becomes the marker particle when it contracts (Algorithm~\ref{alg:hullclose:premark}).

\begin{algorithm}
\caption{Hull Closing Subphase: the Pre-Marker Particle} \label{alg:hullclose:premark}
\begin{algorithmic}[1]
\State \Call{ProcessHullCounters}{ }.
\If {$P$ has no children nor any idle neighbors}
    \State Set $P.\mystate \gets$ \emph{marker}.
    \State Contract tail.
\ElsIf {$P$ has a contracted child $Q$ for which \Call{HandoverIsSafe}{$Q$, $P$}}
    \State Set $P.\mystate \gets$ \emph{marker}.
    \State Perform a pull handover with $Q$.
\EndIf
\end{algorithmic}
\end{algorithm}

The marker particle marks the beginning and end of the leader's hull traversal.
This marker role is passed backwards along the followers as the particles move forward through handovers (Algorithm~\ref{alg:hullclose:mark}).
If the marker particle is ever expanded and has no children in the spanning tree, it knows all particles have joined the hull.
It then generates the $all_{exp}$ token and forwards it towards the leader along expanded particles only.

\begin{algorithm}
\caption{Hull Closing Subphase: the Marker Particle} \label{alg:hullclose:mark}
\begin{algorithmic}[1]
\State \Call{ProcessHullCounters}{ }.
\If {$P$ is holding the termination token}
    \State Set $P.\mystate \gets$ \emph{finished}.
    \State Consume the termination token.
\ElsIf {$P$ is expanded}
    \If {$P$ is holding the $all_{exp}$ token and $P.\parent$ is expanded}
        \State Pass the $all_{exp}$ token to $P.\parent$.
    \EndIf
    \If {$P$ has no children nor any idle neighbors}
        \State Generate the $all_{exp}$ token.
    \ElsIf {$P$ has a contracted child $Q$ for which \Call{HandoverIsSafe}{$Q$, $P$}}
        \State Set $Q.\mystate \gets$ \emph{pre-marker}.
        \State Set $P.\mystate \gets$ \emph{hull}.
        \State Perform a pull handover with $Q$.\vspace{-.6mm}
    \EndIf
\ElsIf {$P$ is contracted and $P.\parent$ is expanded}
    \State Perform a push handover with $P.\parent$.
\EndIf
\end{algorithmic}
\end{algorithm}

Follower particles act exactly as they do in the learning phase, following their parents in the spanning tree through handovers, until they either become the pre-marker or marker particle (Algorithm~\ref{alg:hullclose:follow}).

\begin{algorithm}
\caption{Hull Closing Subphase: Follower Particles} \label{alg:hullclose:follow}
\begin{algorithmic}[1]
\State \Call{ProcessHullCounters}{ }.
\If {$P$ is expanded}
    \If {$P$ has no children nor any idle neighbors}
        \State Contract tail.
    \ElsIf {$P$ has a contracted child $Q$ for which \Call{HandoverIsSafe}{$Q$, $P$}}
        \State Perform a pull handover with $Q$.
    \EndIf {}
\ElsIf {$P.\parent$ is expanded, $P.\parent.\mystate \neq$ \emph{hull} and \Call{HandoverIsSafe}{$P$, $P.\parent$}}
    \If {$P.\parent.\mystate =$ \emph{leader}}
        \State Set $P.\mystate \gets$ \emph{pre-marker}.
    \ElsIf {$P.\parent.\mystate =$ \emph{pre-marker}}
        \State Set $P.\parent.\mystate \gets$ \emph{marker}.
    \ElsIf {$P.\parent.\mystate =$ \emph{marker}}
        \State Set $P.\mystate \gets$ \emph{pre-marker}.
        \State Set $P.\parent.\mystate \gets$ \emph{hull}.
    \EndIf
    \State Perform a push handover with $P.\parent$.
\EndIf
\end{algorithmic}
\end{algorithm}

\subsection{The Hull Filling Subphase} \label{subapp:hullfillcode}

The hull filling subphase is the final phase of the algorithm (Section~\ref{subsec:formationalg}).
This subphase does not use the \emph{pre-marker} and \emph{marker} states, but does use four new ones: \emph{pre-filler}, \emph{filler}, \emph{trapped}, and \emph{pre-finished}.
Only the particles' states and parent pointers are still relevant; all other variables --- including the bits and tokens of the binary counters --- can now be ignored.
One new token type is introduced in this subphase: the ``all-contracted'' token $all_{con}$.
A counter $all_{con}.t \in \{0, \ldots, 6\}$ stores the number of turns this token has been passed through.

The hull was closed when the leader encountered the marker particle.
It begins the hull filling subphase by setting the marker particle as its parent, becoming finished, and generating the $all_{con}$ token (Algorithm~\ref{alg:hullfill:lead}).

\begin{algorithm}
\caption{Hull Filling Subphase: the Leader Particle} \label{alg:hullfill:lead}
\begin{algorithmic}[1]
\State Set $\ell.\parent \gets Q$, where $Q$ is the marker particle.
\State Set $\ell.\mystate \gets$ \emph{finished}.
\State Generate the $all_{con}$ token with $all_{con}.t = 0$.
\end{algorithmic}
\end{algorithm}

The hull particles continue following the leader by performing handovers as in the previous subphase and become finished when their parent is finished (Algorithm~\ref{alg:hullfill:hull}).

\begin{algorithm}
\caption{Hull Filling Subphase: Hull Particles} \label{alg:hullfill:hull}
\begin{algorithmic}[1]
\If {$P.\parent.\mystate =$ \emph{finished}}
    \State Set $P.\mystate \gets$ \emph{finished}.
\ElsIf {$P$ is expanded and has a contracted hull child $Q$}
    \State Perform a pull handover with $Q$.
\ElsIf {$P$ is contracted and $P.\parent$ is an expanded hull particle}
    \State Perform a push handover with $P.\parent$.
\EndIf
\end{algorithmic}
\end{algorithm}

Follower particles continue to follow their parents in the spanning tree through handovers until they encounter a finished particle in their neighborhood (Algorithm~\ref{alg:hullfill:follow}).

\begin{algorithm}
\caption{Hull Filling Subphase: Follower Particles} \label{alg:hullfill:follow}
\begin{algorithmic}[1]
\If {$P$ has a finished neighbor $Q$}
    \State Set $P.\parent \gets$ $Q$, preferring the tail of $Q$ over its head.
\ElsIf {$P$ is expanded}
    \If {$P$ has no children nor any idle neighbors} contract tail.
    \ElsIf {$P$ has a contracted follower child $Q$} perform a pull handover with $Q$.
    \EndIf
\ElsIf {$P.\parent$ is an expanded follower or pre-finished particle}
    \If {$P.\parent.\mystate =$ \emph{pre-finished}}
        \State Set $P.\parent.\mystate \gets$ \emph{finished}.
    \EndIf
    \State Perform a push handover with $P.\parent$.
\EndIf
\end{algorithmic}
\end{algorithm}

Finished particles must inform their follower neighbors whether they are trapped inside or are on the outside of the hull.
A finished particle may also be holding the $all_{con}$ token.
If $all_{con}.t = 7$, then this token has made seven turns, completing a traversal of the hull along all contracted particles, so the hull is completely formed and termination is triggered.
Otherwise, $all_{con}$ is passed to its contracted finished child, updating the turn counter if necessary (Algorithm~\ref{alg:hullfill:finish}).

\begin{algorithm}
\caption{Hull Filling Subphase: Finished Particles} \label{alg:hullfill:finish}
\begin{algorithmic}[1]
\If {$P$ is holding the $all_{con}$ token}
    \If {$all_{con}.t = 7$} consume $all_{con}$ and broadcast termination tokens to all neighbors.
    \ElsIf {$P$ has a contracted finished child $Q$}
        \If {direction $(P.\parent + 3) \text { mod } 6$ does not lead to $Q$} update $all_{con}.t \gets all_{con}.t + 1$.
        \EndIf
        \State Pass the $all_{con}$ token to $Q$.
    \EndIf
\ElsIf {$P.\parent.\mystate =$ \emph{finished} and $P$ has a finished child $Q$}
    \State Let $d = 6$ if $P$ is contracted and $d = 10$ otherwise.
    \State Let $i \in [d]$ be the direction from $P$ to $Q$ and $j \in [d]$ be the direction from $P$ to $P.\parent$.
    \If {$P$ has a contracted follower neighbor $R$ in some direction $k$}
        \State Set $R.\parent \gets P$.
        \If {$k \in \{(i+1) \text{ mod } d, \ldots, (j+d-1) \text{ mod } d\}$} set $R.\mystate \gets$ \emph{filler} and \Return
        \Else {} set $R.\mystate \gets$ \emph{trapped} and \Return
        \EndIf
    \EndIf
\EndIf
\end{algorithmic}
\end{algorithm}

Filler particles perform a clockwise traversal of the outside of the convex hull structure, searching for an expanded finished particle that to handover with (Algorithm~\ref{alg:hullfill:filler}).

\begin{algorithm}
\caption{Hull Filling Subphase: Filler Particles} \label{alg:hullfill:filler}
\begin{algorithmic}[1]
\If {$P$ is expanded}
    \If {$P$ has no children nor any idle neighbors} contract tail.
    \ElsIf {$P$ has a contracted follower child $Q$} perform a pull handover with $Q$.
    \EndIf
\Else {}
    \State Let $i \in [6]$ be such that a finished particle occupies the node in direction $i$.
    \While {a finished particle occupies the node in direction $i$} $i \gets (i + 5) \text{ mod } 6$.
    \EndWhile
    \If {the tail of an expanded particle $Q$ occupies the node in direction $(i + 1) \text{ mod } 6$}
        \State Set $P.\mystate \gets$ \emph{pre-finished}.
        \State Perform a push handover with $Q$ and set $P.\parent \gets Q$.
    \ElsIf {the node in direction $i$ is unoccupied}
        \State Expand in direction $i$.
    \EndIf
\EndIf
\end{algorithmic}
\end{algorithm}

A trapped particle is always contracted and attempts to replace a finished particle in the hull by changing the state of the finished particle to \emph{pre-filler}.
These pre-filler particles leave the hull to become filler particles.
The trapped particle then takes the pre-filler particle's place in the hull, changing its own state to \emph{pre-finished}. (Algorithm~\ref{alg:hullfill:trap}).

\begin{algorithm}
\caption{Hull Filling Subphase: Trapped Particles} \label{alg:hullfill:trap}
\begin{algorithmic}[1]
\If {$P.\parent$ is an expanded finished or pre-filler particle}
    \State Perform a push handover with $P.\parent$.
    \If {$P$ pushed the head of $P.\parent$} set $P.\parent.\parent \gets P$.
    \EndIf
    \State Set $P.\mystate \gets$ \emph{pre-finished}.
    \State Set $P.\parent$ to be the first finished particle clockwise from the current $P.\parent$.
    \If {$P.\parent.\mystate =$ \emph{pre-filler}}
        \State Set $P.\parent.\mystate \gets$ \emph{filler}.
        \If {$P.\parent$ is holding the $all_{con}$ token} take the $all_{con}$ token from $P.\parent$.
        \EndIf
    \EndIf
\ElsIf {$P.\parent$ is a contracted finished particle}
    \State Set $P.\parent.\mystate \gets$ \emph{pre-filler.}
\EndIf
\end{algorithmic}
\end{algorithm}

Pre-filler particles have been marked by a trapped particle for replacement.
Once they successfully leave the hull, they become filler particles (Algorithm~\ref{alg:hullfill:prefiller}).

\begin{algorithm}
\caption{Hull Filling Subphase: Pre-Filler Particles} \label{alg:hullfill:prefiller}
\begin{algorithmic}[1]
\If {$P$ is contracted}
    \State Let $i \gets (P.\parent + 5) \text{ mod } 6$.
    \If {the node in direction $i$ is unoccupied} expand in direction $i$.
    \EndIf
\ElsIf {$P$ has a contracted trapped child $Q$}
    \State Set $P.\mystate \gets$ \emph{filler}.
    \State Set $Q.\mystate \gets$ \emph{pre-finished}.
    \State Perform a pull handover with $Q$.
\EndIf
\end{algorithmic}
\end{algorithm}

Pre-finished particles are always expanded and become finished when they contract, either on their own or by pulling a contracted follower in a handover (Algorithm~\ref{alg:hullfill:prefinish}).

\begin{algorithm}
\caption{Hull Filling Subphase: Pre-Finished Particles} \label{alg:hullfill:prefinish}
\begin{algorithmic}[1]
\If {$P$ has no children nor any idle neighbors}
    \State Set $P.\mystate \gets$ \emph{finished}.
    \State Contract tail.
\ElsIf {$P$ has a contracted follower child $Q$}
    \State Set $P.\mystate \gets$ \emph{finished}.
    \State Perform a pull handover with $Q$.
\EndIf
\end{algorithmic}
\end{algorithm}

\end{document}